\documentclass[12pt]{article}
\usepackage{graphicx}
\usepackage{tikz} %for diagram
\usepackage{amsmath,amsthm,amssymb,mathrsfs,thm-restate,mathtools,bm}
\usepackage[OT1]{fontenc} 
\usepackage{etoc}
\usepackage{float}
\usepackage{titlesec}
\usepackage{multirow,array}
\usepackage{diagbox}
\usepackage{fancyhdr}
\titleformat{\section}
		{\scshape \large \bfseries\center}
         {\thesection}% the label and number
        {0.5em}% space between label/number and subsection title
        {}% formatting commands applied just to subsection title
        []% punctuation or other commands following subsection title
\titleformat{\subsection}%[runin] runin puts it in the same paragraph
        {\normalfont\bfseries}% formatting commands to apply to the whole heading
        {\thesubsection}% the label and number
        {0.5em}% space between label/number and subsection title
        {}% formatting commands applied just to subsection title
        []% punctuation or other commands following subsection title

\allowdisplaybreaks

\def\subsectiontitle{}
\def\subsubsectiontitle{}
\usepackage{comment}
\usepackage[backgroundcolor=orange!50!white]{todonotes}
\usepackage{tikz}
\usetikzlibrary{automata, positioning}

\usepackage{kpfonts}
\usepackage{inconsolata}

\usepackage{makecell}
\usepackage{cellspace}
\usepackage{enumitem}
\usepackage[margin=1in]{geometry}
\usepackage[colorlinks,breaklinks=true]{hyperref}
\usepackage{breakcites}
\usepackage{xcolor}
\AtBeginDocument{%
	\hypersetup{%
    linkcolor={red!50!black},%
    citecolor={blue!50!black},%
    urlcolor={blue!80!black}%
}%
}%

\usepackage[nameinlink,noabbrev,sort,capitalise]{cleveref}
\makeatletter
\def\ps@pprintTitle{%
 \let\@oddhead\@empty
 \let\@evenhead\@empty
 \def\@oddfoot{\emph{Very preliminary version}\hfill\emph{This draft: \today}}%
 \let\@evenfoot\@oddfoot}
\makeatother
\usepackage{etoolbox}
\patchcmd{\pprintMaketitle}
  {\fi\hrule}% the second rule
  {\fi\ifvoid\extrainfobox\else\unvbox\extrainfobox\par\vskip10pt\fi\hrule}
  {}{}

\newsavebox\extrainfobox

\newtheorem{proposition}{Proposition}
\crefname{proposition}{Proposition}{Propositions}

\newtheorem{theorem}{Theorem}
\crefname{theorem}{Theorem}{Theorems}

\crefname{corollary}{Corollary}{Corollaries}

\newtheorem{lemma}{Lemma}
\crefname{lemma}{Lemma}{Lemmas}

\crefname{assumption}{Assumption}{Assumptions}

\crefname{axiom}{Axiom}{Axioms}

\newtheorem{definition}{Definition}
\crefname{definition}{Definition}{Definitions}

\theoremstyle{remark}
\newtheorem*{remark}{Remark}

\theoremstyle{definition}
\newtheorem{example}{Example}
\crefname{example}{Example}{Examples}

\crefname{problem}{Problem}{Problems}

\theoremstyle{claim}
\newtheorem{claim}{Claim}
\crefname{claim}{Claim}{Claimss}

\usepackage{indentfirst}

\allowdisplaybreaks

\let\oldfootnote\footnote
\renewcommand\footnote[1]{\oldfootnote{\hspace{.4mm}#1}}
\makeatletter
\renewenvironment{proof}[1][\proofname] {\par\pushQED{\qed}\normalfont\topsep6\p@\@plus6\p@\relax\trivlist\item[\hskip\labelsep\bfseries#1\@addpunct{.}]\ignorespaces}{\popQED\endtrivlist\@endpefalse}
\makeatother
\let\oldFootnote\footnote
\newcommand\nextToken\relax

\renewcommand\footnote[1]{%
    \oldFootnote{#1}\futurelet\nextToken\isFootnote}

\newcommand\isFootnote{%
    \ifx\footnote\nextToken\textsuperscript{,}\fi}
\usepackage[american]{babel}
\usepackage[authoryear,round]{natbib}
\def\R{\mathcal{R}}

\usepackage{todonotes}

\usepackage{amstext} % for \text macro
\usepackage{array}   % for \newcolumntype macro
\newcolumntype{L}{>{$}l<{$}} % math-mode version of "l" column type

\newcommand{\hbz}{\hat{\mathbf{z}}}
\newcommand{\hbbz}{\hat{\mathbf{Z}}}

\begin{document}

\title{Individual Sovereignty and Other-Regarding Preferences}

\date{This version: \today}

\author{\makebox[.25\linewidth]{{John Mori}\thanks{University of Chicago Department of Economics; email: \protect\texttt{johnmori@uchicago.edu}. I am grateful to Ben Brooks, Phil Reny, and Joe Root for guidance and feedback. I also thank the participants at EC’26, SSCW 2026, the 2026 Stony Brook International Conference on Game Theory, and the student theory workshop at the University of Chicago for helpful comments.}}}

\maketitle

\begin{abstract}
    We consider the social aggregation of preferences over lotteries in the presence of other-regarding preferences. If society respects each individual's \emph{sovereignty}, an axiom we propose akin to Sen's \emph{Liberalism}, then society's utility is a linear combination of individuals' self-regarding utilities. That is, other-regarding preferences can only influence the weights society places on each individual. We next characterize the unique weighting method under which society's weight ratio between two individuals is the geometric mean of those of all individuals. The central axiom requires that when only a group of individuals is at stake, society's preferences depend only on individuals' preferences towards that group. We extend the first result to a setting with feasibility constraints and to a setting with subjective uncertainty.
\end{abstract}

\section{Introduction}

People vote not only according to their economic self-interest, but also according to their \emph{other-regarding preferences} \citep{blais2000vote}. Some even favor policies against their economic self-interest. For example, a wealthy person concerned about fairness might support a policy implementing a more progressive tax system than the status quo. Should society deem this person better off under the more progressive tax system? On the one hand, they prefer the policy; on the other, paying higher taxes would reduce their consumption. In this paper, we tackle the question of how society should weigh other-regarding preferences in a social welfare analysis.

If other-regarding preferences are included in social welfare, society may adopt policies that are supported primarily by people's other-regarding preferences but are \emph{paternalistic} --- they require individuals to take, or prohibit them from taking, actions that affect only themselves, possibly against their own preferences. Governments, to varying degrees, require saving for retirement, restrict risky investment products, mandate various safety measures, and restrict forms of gambling and drug use. One could argue against such policies on the principle that society should respect people's individual sovereignty --- each person's right to make decisions when only their life, property, or body is affected.

In this paper, we demonstrate a tension between society weighing individuals' other-regarding preferences and respecting individual sovereignty. We then provide a theory of interpersonal comparisons of utility derived from individuals' other-regarding preferences.

We model other-regarding preferences in the following way: each individual $i$ has a set of individual outcomes $X_i$, and the set of social outcomes is the Cartesian product of sets of individual outcomes. As in \cite{harsanyi1955cardinal}, our setting involves risk --- individuals and society have expected utility preferences over lotteries over the set of social outcomes. Each individual's preferences comprise their self-regarding and other-regarding preferences.

We introduce a condition on social preferences called \emph{individual sovereignty}. Social preferences satisfy individual sovereignty if, for comparisons of lotteries in which only individual $i$ is at stake (i.e. only the marginal distribution over $i$'s outcome differs), they defer to $i$'s preferences. Theorem \ref{theorem1} shows that this condition, in conjunction with an independence assumption on individual preferences, implies that society's utility is a linear combination of individuals' \emph{self-regarding} utilities --- their utilities over their own outcomes. We call this class of social preferences \emph{self-regarding utilitarian}.

We then turn to the question of making interpersonal comparisons of utility, that is, how a self-regarding utilitarian society places weights on each individual's self-regarding utility. We propose an approach under which individuals' other-regarding preferences serve as a foundation for interpersonal comparisons. Theorem \ref{theoremgeom} characterizes a particular social welfare function we call \emph{geometric} self-regarding utilitarianism, under which society's weight ratio between any two individuals is the geometric mean of those of all individuals. The central axiom, \emph{independence of outside concerns}, requires that when only a group of individuals is at stake, society's preferences depend only on individuals' preferences towards that group.

We provide two extensions of Theorem \ref{theorem1}. The first extension is to a setting where the set of social outcomes is constrained to a subset of the Cartesian product of individual outcomes. Most applications of our results have a feasible set of social outcomes that is a strict subset of the Cartesian product. Theorem \ref{theorem3} establishes an analogue of Theorem \ref{theorem1} in this setting, and we find that the structure of the subset determines the strength of individual sovereignty. We discuss various design settings in relation to this observation.

Finally, the second extension is to a Savage setting, where individuals have heterogeneous subjective beliefs. Theorem \ref{theorem2} shows that a modified sovereignty condition implies that, similarly, society's utility is a linear combination of individuals' self-regarding utilities, and also that society's subjective belief is an affine combination of individuals' beliefs.

\textbf{Related literature}. Our sovereignty condition harks back to the condition \emph{Liberalism} introduced by \cite{sen1970impossibility}. A social welfare function satisfies Liberalism if for each individual, there exists a pair of alternatives for which the social preference defers to the individual preference. Sen demonstrates an impossibility --- later dubbed the ``Liberal Paradox'' --- between the conditions Liberalism, weak Pareto, and unrestricted domain.

Sen and much of the literature that followed (see \cite{suzumura2011welfarism} for a review) do not decompose the set of social outcomes as a Cartesian product of individual outcomes, as we do. Several authors responding to Sen do examine such a product structure --- \cite{gibbard1974pareto} poses another impossibility result, \cite{hammond2002rights} examines liberalism in an exchange economy, \cite{breton1999separable} examine strategyproof mechanisms, and \cite{hammond1982liberalism}, \cite{igersheim2013invoking}, and \cite{dougherty2022effect} make restrictions on the domain of preferences to avoid the impossibility result. We instead eschew weak Pareto. Many have argued that weak Pareto is less compelling in a setting with other-regarding preferences (see \cite{suzumura2011welfarism}).

Sovereignty is also closely related to the \emph{strong principle of acceptance} of \cite{karni1998informationally}, who present a version of Harsanyi's impartial observer theorem (\cite{harsanyi1953cardinal}, \cite{harsanyi1955cardinal}, \cite{harsanyi1977rational}). Their observer, who has an equal chance of being each individual, ranks two prospects that differ only in the lottery faced by individual $i$ in the same way as $i$ does, and they show that the observer's preferences are represented by a sum of individuals' utilities. The difference from our setting lies in what individual $i$'s lottery is over. In their setting, it is over entire social alternatives, and so is ranked by $i$'s full preferences, including any other-regarding preferences. In ours, it is over $i$'s own outcomes, and so sovereignty excludes other-regarding preferences from social preferences.

Our results also contribute to the literature on other-regarding preferences. Much has been written on the experimental evidence for other-regarding preferences and their role in games and markets (see \cite{sobel2005interdependent} for a review). Much less attention has been given to other-regarding preferences in the social welfare analysis we pursue (see \cite{fleurbaey2009beyond} for a discussion). A few recent papers have begun to include other-regarding preferences in welfare analysis for settings without uncertainty (\cite{decerf2016fair}, \cite{treibich2019welfare}, \cite{fleurbaey2021fair}, \cite{segal2026aggregating}).

We model individuals as having other-regarding preferences over the \emph{outcomes} of other individuals, known in the literature as \emph{paternalistic} other-regarding preferences \citep{jones1992paternalistic}. Individual $i$ can prefer an outcome for individual $j$ that $j$ themself does not prefer. This is in contrast with the more commonly assumed \emph{non-paternalistic} other-regarding preferences --- those that are a function of other individuals' \emph{preferences} (see \cite{sobel2005interdependent} for many examples). We make this modeling choice for its generality: it can capture both paternalism and a form of non-paternalistic altruism we introduce in Section \ref{geometric}. There is experimental evidence for such paternalistic preferences (see \cite{ambuehl2021motivates} and the references therein). For instance, \cite{jacobsson2007altruism} find that individuals are paternalistic with regard to other individuals' health, preferring to give others in-kind donations of a health treatment rather than cash. Also, individual $i$ might agree with $j$'s preferences over $j$'s own outcomes ordinally, but may display different risk preferences from $j$ (see \cite{batteux2019our} for a review and meta-analysis). For instance, in many health settings, individuals are more risk-averse when making decisions on behalf of others.

Recently, \cite{eden2025ethical} consider a model where each individual's utility is a sum of their self-regarding utility --- defined over their own consumption --- and their (paternalistic) other-regarding utility over the consumption profile of society. They show that any social welfare function that is increasing in the utility of each individual is nearly entirely determined by the individuals' other-regarding utilities. Their results complement ours --- if the social welfare function is Paretian, as they assume, then it largely ignores individuals' self-regarding preferences. On the other hand, if the social welfare function is liberal, as we assume, then it largely ignores other-regarding preferences.

Whether other-regarding preferences \emph{should} be included in welfare analysis is debated. \cite{milgrom1993sympathy} argues against their inclusion --- if individual $i$ cares about individual $j$'s consumption, $j$'s consumption would be ``double counted'' by a welfare analysis that weighs other-regarding preferences (see also \cite{dworkin1990double}). Others have given conditions in various settings under which an analysis including other-regarding preferences yields the same conclusions as one with only self-interested individuals (\cite{winter1969simple}, \cite{dufwenberg2011other}, \cite{eden2024irrelevance}). Theorem \ref{theorem1} limits the role of other-regarding preferences to just determining weights, while Theorem \ref{theoremgeom} prescribes a particular weighting as a function of other-regarding preferences.

Finally, Theorem \ref{theoremgeom} speaks to the literature on making interpersonal comparisons of utility.\footnote{See \cite{sen1999possibility}, \cite{d2002interpersonal}, and \cite{fleurbaey2011theory} for reviews of the literature.} In our shared setting with objective uncertainty, \cite{harsanyi1955cardinal} characterized utilitarianism, but left open the question of how to assign weights. \emph{Relative utilitarianism}, which sums individuals' vNM utilities normalized to the unit interval, has been characterized by many authors (\cite{karni1998impartiality}, \cite{dhillon1999relative}, \cite{segal2000let}, \cite{borgers2017revealed}). Our approach differs in that geometric self-regarding utilitarianism does not satisfy Harsanyi's Pareto condition.

Section \ref{objective} presents the model in a setting with objective uncertainty and gives a characterization of self-regarding utilitarianism. Section \ref{geometric} introduces additional axioms to characterize geometric self-regarding utilitarianism. Section \ref{constrain} extends Theorem \ref{theorem1} to a setting where social outcomes are constrained to a subset of the Cartesian product. Section \ref{subjective} extends Theorem \ref{theorem1} to a setting with subjective uncertainty. Section \ref{conclusion} concludes. All proofs are in Appendix \ref{proofs}.

\section{Objective Uncertainty} \label{objective}

Let $I$ be a set of $n$ individuals, with generic elements $i$ and $j$. For each individual $i \in I$, let $X_i$ be a set of \emph{individual} outcomes for individual $i$, with generic element $x_i$. Examples of individual outcomes include consumption bundles, assignments or allocations, and individual experiences of a social policy. We assume that each $X_i$ contains at least two outcomes.

The set of \emph{social} outcomes $X$ is the $n$-fold Cartesian product of sets of individual outcomes, i.e. $X = X_1 \times \dots \times X_n$, where a social outcome $x = (x_1, \dots, x_n)$ specifies an individual outcome for each individual. For any $i \in I$, denote by $X_{-i} = X_1 \times \dots \times X_{i-1} \times X_{i+1} \times \dots \times X_n$ the $(n-1)$-fold Cartesian product of sets of outcomes for individuals other than $i$.

Denote by $\Delta(X)$ the set of probability distributions with finite support, or lotteries, over $X$, with generic elements $P$ and $Q$. For each $i \in I$, $\Delta(X_i)$ is the set of lotteries over $X_i$. For any lottery $P \in \Delta(X)$, denote by $P_i$ the marginal distribution on $X_i$, and by $P_{-i}$ the marginal distribution on $X_{-i}$.

Each individual $i \in I$ has preferences $\succsim_i \subseteq \Delta(X) \times \Delta(X)$. We assume that each individual's preferences admit a von Neumann-Morgenstern expected utility (EU) representation --- say that $\succsim_i$ is \emph{vNM-rational} if there exists a bounded utility index $u_i: X \rightarrow \mathbb{R}$ such that for any lottery $P \in \Delta(X)$, the \emph{EU representation} $\int_X u_i(x)dP(x)$ represents $\succsim_i$.\footnote{We use the integral notation for simplicity, where $\int_X u_i(x)dP(x) := \sum_{x \in \text{supp}(P)}u_i(x)P(x)$.} Let $\R$ denote the set of vNM-rational preferences over $\Delta(X)$.

So far, this setting allows individuals to care about correlations in outcomes across individuals, i.e. to have strict preferences between lotteries $P$ and $Q$ that have the same marginal distribution for each individual but differ in their joint distribution. In particular, individuals may have ex-post distributional concerns, such as a preference for ex-post fairness.\footnote{See \cite{grant2012equally} for a characterization of Harsanyi's impartial observer with a preference for ex-post fairness.} We limit the scope of these distributional concerns with a condition adapted from \cite{fishburn1965independence} ---  we still allow individuals to care about correlations between the outcomes of other individuals, but we assume that each individual does not care about correlations between \emph{their own} outcome and those of other individuals.

\begin{definition}
    $\succsim_i$ satisfies \emph{mutual independence of $X_i$ and $X_{-i}$} if for all $P, Q \in \Delta(X)$ for which $P_i = Q_i$ and $P_{-i} = Q_{-i}$, $P \sim_i Q$.
\end{definition}

An individual whose preferences satisfy mutual independence may have a concern for ex-post fairness when considering other individuals, but not when considering themself among other individuals. Note that mutual independence does \emph{not} require individuals to be purely selfish --- they may still have rich other-regarding preferences.

\cite{fishburn1965independence} shows that mutual independence of $X_i$ and $X_{-i}$ is equivalent to the existence of an additive utility index.\footnote{Fishburn's setting is not specific to other-regarding preferences. \cite{simon2016existence} applies Fishburn's result to other-regarding preferences.}

\begin{proposition} \label{prop1}
    \citep{fishburn1965independence} $\succsim_i$ satisfies mutual independence of $X_i$ and $X_{-i}$ if and only if there exist \emph{subutility indices }$u_{i,i}: X_i \rightarrow \mathbb{R}$ and $u_{i,-i}: X_{-i} \rightarrow \mathbb{R}$ such that $\succsim_i$ is represented by $u_i(x) = u_{i,i}(x_i) + u_{i,-i}(x_{-i})$ for any $x \in X$. Furthermore, for a given utility index $u_i$, subutility indices $u_{i,i}$ and $u_{i,-i}$ are unique up to simultaneous transformations $u_{i,i}' = u_{i,i} + c$ and $u_{i,-i}' = u_{i,-i} - c$ for any $c \in \mathbb{R}$.
\end{proposition}

We interpret each individual $i$'s subutility index $u_{i,i}$ as capturing their \emph{self-regarding} preferences over lotteries over their own outcomes in $X_i$, and $u_{i,-i}$ as capturing their \emph{other-regarding} preferences over lotteries over the outcomes of others in $X_{-i}$. Let $\succsim_{i,i}$ be the self-regarding preferences over $\Delta(X_i)$ induced by the self-regarding utility $u_{i,i}$. This representation aligns with the additive utility functions commonly assumed in the literature on other-regarding preferences \citep{sobel2005interdependent}.

Mutual independence merely requires indifference between lotteries with shared marginal distributions, but, combined with vNM-rationality, it implies that an individual's self-regarding preferences are separable from their other-regarding preferences \emph{for all lotteries}.

We now turn to social preferences. Let $\succsim_0 \subseteq \Delta(X) \times \Delta(X)$ be society's preferences. In the spirit of \cite{harsanyi1955cardinal}, we assume that society's preferences are vNM-rational, with utility index $u_0$. We now introduce our first condition for social aggregation:

\begin{definition}
    Social preferences $\succsim_0$ satisfy \emph{(individual) sovereignty} if for any $i \in I$ and any $P, Q \in \Delta(X)$ where $P_{-i} = Q_{-i}$, $P \succsim_0 Q$ if and only if $P \succsim_i Q$.
\end{definition}

That is, between two lotteries for which only individual $i$'s ex-ante outcome (potentially) differs, society's preference should defer to individual $i$'s preference.

Note that not only do lotteries $P$ and $Q$ share the same marginal distribution over outcomes for each individual $j$ other than $i$ ($P_j = Q_j$ for all $j \neq i$), but they also share the same joint distribution over $X_{-i}$. So if individual $i$'s preferences satisfy mutual independence, their other-regarding preferences must be indifferent between $P_{-i}$ and $Q_{-i}$, even if individual $i$ cares about ex-post fairness among other individuals. Thus only $i$'s self-regarding preferences determine $i$'s preference between $P$ and $Q$.

Say that social preferences $\succsim_0$ are \emph{(ex-ante) paternalistic} if $\succsim_0$ does not satisfy sovereignty. That is, for some individual $i$, there exist two lotteries between which only $i$'s ex-ante outcome differs, and, if given a choice between the two, society and $i$ would choose differently. Note that society can agree with individual $i$ on the preference order of \emph{ex-post} outcomes for $i$ and yet still be ex-ante paternalistic. For example, both a car passenger and society prefer that the passenger avoid serious injury in an accident, but society might impose a more risk-averse attitude on the passenger and require them to wear a seatbelt.

Loosely, sovereignty is a strengthening of \citeauthor{sen1970impossibility}'s \citeyearpar{sen1970impossibility} Liberalism, which requires only that there exists a pair of alternatives for each individual over which they have dictatorial power.\footnote{Sen's Liberalism is defined in a social choice setting without uncertainty, where the social outcome is not expressed as a profile of individual outcomes.} Sovereignty requires that each individual $i$ dictate the social preference between all pairs of lotteries that hold the marginal outcomes of others fixed. The strength of sovereignty pins down the functional form of the social utility index:

\begin{theorem} \label{theorem1}
    Suppose each $\succsim_i$ satisfies mutual independence of $X_i$ and $X_{-i}$. $\succsim_0$ satisfies individual sovereignty if and only if $u_0$ is a positive linear combination of $\{u_{i,i}\}_{i \in I}$.
\end{theorem}

Say that such $\succsim_0$ is \emph{self-regarding utilitarian}, as the social utility index is a weighted sum of individuals' self-regarding utilities.\footnote{\cite{breton1999separable} examine a social \emph{choice} function analogous to self-regarding utilitarianism. In a similar setting where social outcomes are a Cartesian product of individual outcomes, they show that a social choice function is strategyproof if and only if it is \emph{libertarian}, i.e. every individual dictates their own outcome.}

We began in a setting where individuals can have other-regarding preferences, so long as these other-regarding preferences are mutually independent of their self-regarding preferences. Sovereignty requires that society ignore all other-regarding preferences when choosing between lotteries where only individual $i$'s ex-ante outcome differs. But it turns out that, given the vNM-rationality of society, society must also ignore all other-regarding preferences when choosing between any pair of lotteries, except in determining the weights it places on each self-regarding utility.

For those who are sympathetic to excluding other-regarding preferences from welfare analysis, sovereignty offers a normative principle for doing so. For instance, it is standard practice to model an individual's utility over just their own consumption. Also, in allocation and matching settings, standard mechanisms elicit reports that contain only preferences over an individual's allocation.\footnote{In some settings, such as school choice, mutual independence of self-regarding and other-regarding preferences may not hold, since a student's preferences over schools could be determined by the eventual peers at their assigned school.}

For those who are sympathetic to including other-regarding preferences in social decision making and welfare analysis beyond the determination of weights, Theorem \ref{theorem1} provides an argument against such a strong notion of liberalism as sovereignty.

\section{Geometric self-regarding utilitarianism} \label{geometric}

We now turn to the question of how society should make interpersonal comparisons of utility, i.e. determine the weights it places on each self-regarding utility. Our approach begins with the  observation that individuals already implicitly make interpersonal comparisons when they knowingly take actions that affect others, and particularly when they take altruistic actions at their own expense. For instance, any donation to a charity could have been donated to another charity. Supporting a policy with heterogeneous effects across a population is implicitly founded on some cost-benefit analysis. Society can take these other-regarding preferences as input into how it should make interpersonal comparisons.

In Section \ref{objective} we assumed that each individual $i$ satisfied mutual independence of $X_i$ and $X_{-i}$. In this section we further assume that each individual is self-regarding utilitarian \emph{themself}. Formally, a preference profile $\succsim_I$ is self-regarding utilitarian if there exist subutility indices $\{u_{j,j}\}_{j \in I}$ such that each $\succsim_i$ is represented by a utility index $u_i$ that is a positive linear combination of $\{u_{j,j}\}_{j \in I}$. Then for each $j \in I$, $u_{j,j}$ is $j$'s self-regarding utility, unique up to a positive affine transformation. We call $\{u_{j,j}\}_{j \in I}$ the \emph{self-regarding utilities} of $\succsim_I$. That is, each individual's concern for $j$ agrees with $j$'s own self-regarding preferences $\succsim_{j,j}$. Say that a preference $\succsim$ is \emph{self-regarding utilitarian with respect to} $\succsim_I$ if $\succsim$ is represented by a positive linear combination of the self-regarding utilities of $\succsim_I$.

While there is precedent in the literature for assuming that individuals are altruistic when voting and making social decisions, our assumption is strong.\footnote{From a normative angle, \cite{harsanyi1977morality} argues that individuals should adopt the viewpoint of an impartial observer when making social decisions --- particularly a utilitarian one. \cite{feddersen2006theory} appeal to the existence of ethical voters to explain voter turnout. \cite{flanigan2023distortion} show that the existence of public-spirited voters reduces the distortion --- the cost of aggregating ordinal revealed preferences with respect to underlying cardinal preferences --- of social choice rules.} Each individual displays some degree of altruism towards \emph{every} other individual.  They may differ on the weights they place on others, which realistically could be quite small.\footnote{We occasionally refer to ``weights'' that individuals place on each other's self-regarding utility. Since we do not assume a particular normalization of self-regarding utilities, direct comparisons of weights are meaningless. For instance, it is meaningless to say that $i$ places a higher weight on $j$'s self-regarding utility than on $k$'s, since a different self-regarding utility representation could flip that assessment. On the other hand, comparisons between \emph{comparisons} of weights are meaningful. For instance, it is meaningful to say that $i$'s weight ratio of $j$'s to $k$'s self-regarding utility is higher than $\ell$'s weight ratio of $j$ to $k$.} Also, $i$'s altruism towards $j$, for risky outcomes, displays the same risk attitudes as $j$'s self-regarding preferences.\footnote{Note that this is a form of non-paternalistic altruism with regard to $j$'s self-regarding utility $u_{j,j}$, instead of their $u_{j}$ as assumed in other literature.} One interpretation applies the same normative argument from the previous section to individuals: each individual, when making decisions that affect others, does not want to infringe on any other individual's sovereignty.

This assumption ensures that each individual $i$ has a \emph{rate of interpersonal substitution} between the self-regarding utility of any individual $j$ and that of any other individual $k$. Informally, given an improvement in $j$'s outcome and an improvement in $k$'s outcome, $i$'s rate of substitution between them is the value $i$ places on $j$'s improvement relative to the value $i$ places on $k$'s improvement. For instance, a rate of 2 means that, with everyone else's outcomes fixed, $i$ is indifferent between $k$'s improvement for sure and a one-half chance of $j$'s improvement. We define rates of substitution formally in Section \ref{pf.geom} of the appendix. A self-regarding utilitarian social preference has these rates of substitution as well. Geometric self-regarding utilitarian social preferences, to determine their rate of substitution between $j$ and $k$, take the \emph{geometric mean} of the rates of substitution between $j$ and $k$ across all individuals $i$.

Our earlier axiom, individual sovereignty, is a \emph{single-profile axiom}, which restricts the output of the social welfare function on each preference profile independently. Pareto properties, which feature in versions of Harsanyi's utilitarianism, are also single-profile axioms. To characterize geometric self-regarding utilitarianism, we require several multi-profile axioms, each of which restricts the output of the social welfare function with a given input preference profile $\succsim_I$ in relation to its output with a different, related profile $\succsim_I'$. Many characterizations of other social welfare functions that determine interpersonal comparisons of utility, e.g. relative utilitarianism, similarly rely on multi-profile axioms (see papers in the related literature section). To that end we introduce a social welfare function.

Let $\mathcal{D}$ be the set of self-regarding utilitarian preference profiles, and let $\Phi: \mathcal{D} \rightarrow \mathcal{R}$ be a social welfare function. Say that $\Phi$ satisfies individual sovereignty if for each $\succsim_I \in \mathcal{D}$, social preferences $\Phi(\succsim_I)$ satisfy individual sovereignty with respect to $\succsim_I$. Since each individual's preferences satisfy mutual independence, Theorem \ref{theorem1} implies that $\Phi$ satisfies individual sovereignty if and only if for each $\succsim_I \in \mathcal{D}$, $\Phi(\succsim_I)$ is self-regarding utilitarian with respect to $\succsim_I$.

Now we present geometric self-regarding utilitarianism. Consider any $\succsim_I \in \mathcal{D}$. Fix any self-regarding utilities $\{u_{j,j}\}_{j \in I}$ of $\succsim_I$, and express each $u_i$ as $u_i(x) = \sum_{j \in I} \alpha_{i,j} u_{j,j}(x_j)$, where each $\alpha_{i,j} > 0$. $\Phi: \mathcal{D} \rightarrow \mathcal{R}$ is \emph{geometric self-regarding utilitarian} if the social utility index of $\Phi(\succsim_I)$ has the following form: $u_0(x) = \sum_{j \in I} \left(\prod_{i\in I}\alpha_{i,j}\right)^{\frac{1}{n}}u_{j,j}(x_j)$. Note that the definition of geometric self-regarding utilitarian preferences is invariant to the choice of self-regarding utilities $\{u_{j,j}\}_{j \in I}$ and utilities $\{u_i\}_{i \in I}$.\footnote{Recall that for each $j \in I$, $u_{j,j}$ is unique up to a positive affine transformation. Consider a particular individual $k \in I$ and $u_{k,k}' = b u_{k,k} + c$ where $b > 0$. Then, for each $i$, $\succsim_i$ is represented by $\alpha_{i,k}' u_{k,k}'(x_k) + \sum_{j \neq k} \alpha_{i,j} u_{j,j}(x_j)$, where $\alpha_{i,k}' = \frac{\alpha_{i,k}}{b}$, since this expression equals $u_i(x) + \frac{\alpha_{i,k}c}{b}$. The social utility index is then

\begin{align*}
    u_0'(x) & = \left(\prod_{i\in I}\alpha_{i,k}'\right)^{\frac{1}{n}}u_{k,k}'(x_k) + \sum_{j \neq k} \left(\prod_{i\in I}\alpha_{i,j}\right)^{\frac{1}{n}}u_{j,j}(x_j) \\
    & = \left(\prod_{i\in I}\frac{\alpha_{i,k}}{b}\right)^{\frac{1}{n}}(bu_{k,k}(x_k) + c) + \sum_{j \neq k} \left(\prod_{i\in I}\alpha_{i,j}\right)^{\frac{1}{n}}u_{j,j}(x_j)
\end{align*}

The $b$'s cancel, so $u_0'$ differs from $u_0$ only by the constant term, and so represents the same preferences.

Similarly, for each $i \in I$, $u_i$ is unique up to positive affine transformations. Consider a particular individual $k \in I$ and let $u_k' = b u_k + c$, where $b > 0$. Then $\succsim_k$ is represented by $\sum_{j \in I} \alpha_{k,j}' u_{j,j}(x_j)$, where $\alpha_{k,j}' = b \alpha_{k,j}$ for all $j \in I$, since this expression equals $u_k'(x) - c$. So the social utility index becomes $u_0' = b^{\frac{1}{n}} u_0$, which represents the same preferences.}

Now we present a characterization. The central axiom, independence of outside concerns, requires society --- when making decisions that affect only a subset of individuals --- to ignore preferences over the outcomes of individuals who are not affected. For any nonempty $J \subseteq I$, let $\mathcal{S}_J = \{(P, Q) \in \Delta(X) \times \Delta(X): P_{-J} = Q_{-J}\}$ be the set of comparisons in which only individuals in $J$ are at stake, where $X_{-J} = \prod_{\ell \in I \setminus J} X_\ell$ is the set of profiles of outcomes for individuals outside $J$, and $P_{-J}$ is the marginal distribution of $P$ on $X_{-J}$. Say that two preferences $\succsim, \succsim' \in \mathcal{R}$ \emph{agree on} $\mathcal{S}_J$ if $P \succsim Q \iff P \succsim' Q$ for all $(P, Q) \in \mathcal{S}_J$. For instance, individual sovereignty requires that social preferences agree with each individual $i$'s preferences on $\mathcal{S}_{\{i\}}$.

\begin{definition}
    $\Phi$ satisfies \emph{independence of outside concerns} if for any $\succsim_I, \succsim_I' \in \mathcal{D}$ and nonempty $J \subseteq I$ such that $\succsim_i$ and $\succsim_i'$ agree on $\mathcal{S}_J$ for all $i \in I$, $\Phi(\succsim_I)$ and $\Phi(\succsim_I')$ agree on $\mathcal{S}_J$.
\end{definition}

That is, when only individuals in $J$ are at stake, society's preferences depend only on individuals' preferences towards $J$. In particular, they do not depend on the self-regarding preferences of individuals outside $J$, or on anyone's concern for them. For instance, if each individual's outcome is the school they attend, and a student outside $J$ changes their preferences over their own school, then society's comparisons of school assignments for the students in $J$ should not change. The axiom has the form of Arrow's independence of irrelevant alternatives, applied to the comparisons in $\mathcal{S}_J$.

Given individual sovereignty, the condition for all groups follows from the condition for pairs $J = \{j, k\}$, since a self-regarding utilitarian preference's comparisons in $\mathcal{S}_J$ are determined by its comparisons in $\mathcal{S}_{\{j,k\}}$ for $j, k \in J$. For a pair, individual $i$'s comparisons in $\mathcal{S}_{\{j,k\}}$ capture exactly $j$'s and $k$'s self-regarding preferences together with $i$'s rates of substitution between improvements in $j$'s and $k$'s outcomes. So between profiles with the same self-regarding preferences, independence of outside concerns requires that if no individual changes their rates of substitution between $j$ and $k$, then neither does society. That is, society's comparison of an improvement in $j$'s outcome with an improvement in $k$'s outcome depends only on how individuals compare these two improvements.

Geometric self-regarding utilitarianism satisfies independence of outside concerns, since its comparison between $j$'s and $k$'s improvements is determined by the ratio of its weights, which is a function of each individual $i$'s respective weight ratio:
\begin{align*}
    \frac{\left(\prod_{i\in I}\alpha_{i,j}\right)^{\frac{1}{n}}}{\left(\prod_{i\in I}\alpha_{i,k}\right)^{\frac{1}{n}}} = \left(\prod_{i\in I}\frac{\alpha_{i,j}}{\alpha_{i,k}}\right)^{\frac{1}{n}}
\end{align*}

The remaining axioms are standard. Unanimity requires that if every individual has the same preferences, then society should also have those preferences.

\begin{definition}
    $\Phi$ is unanimous if for any $\succsim$ such that $(\succsim, \dots, \succsim) \in \mathcal{D}$, $\Phi(\succsim, \dots, \succsim) = \succsim$.
\end{definition}

Note that strong Pareto implies unanimity.\footnote{$\Phi$ satisfies strong Pareto if for any $\succsim_I \in \mathcal{D}$ and any $P, Q \in \Delta(X)$, (i) if $P\succsim_i Q$ for all $i \in I$, then $P\succsim_{\Phi(\succsim_I)} Q$ and (ii) if $P\succsim_i Q$ for all $i \in I$ and $P\succ_j Q$ for some $j \in I$, then $P \succ_{\Phi(\succsim_I)} Q$.} In fact, variants of Harsanyi's aggregation theorem tell us that strong Pareto implies utilitarianism --- that the social utility is a positive linear combination of individuals' utilities (as opposed to their self-regarding utilities; see \cite{weymark1993harsanyi} and \cite{de1995note}). Thus geometric self-regarding utilitarianism, a distinct theory, does not satisfy strong Pareto.\footnote{To see this, suppose $n=3$. For fixed self-regarding utilities $\{u_{i,i}\}_{i = 1,2,3}$, suppose $u_1 = u_{1,1} + u_{2,2} + u_{3,3}$, $u_2 = u_{1,1} + u_{2,2} + \frac{1}{8}u_{3,3}$, and $u_3 = u_{1,1} + \frac{1}{8}u_{2,2} + u_{3,3}$. The corresponding geometric self-regarding utilitarian utility index is $u_0 = u_{1,1} + \frac{1}{2}u_{2,2} + \frac{1}{2}u_{3,3}$, which cannot be expressed as a non-negative linear combination of $u_1, u_2$, and $u_3$. By Harsanyi's aggregation theorem, it follows that social preferences do not satisfy Pareto weak preference, and hence do not satisfy strong Pareto either.}

Anonymity requires that the social welfare function disregard which individual holds which preferences. For any profile $\succsim_I$ and permutation $\sigma: I \rightarrow I$, let $\succsim_I^\sigma$ be the profile where $\succsim_{\sigma(i)}^\sigma\, = \,\succsim_i$ for all $i \in I$. That is, individuals exchange their preferences over lotteries, while each individual's outcome stays their own.

\begin{definition}
    $\Phi$ is anonymous if for any $\succsim_I \in \mathcal{D}$ and permutation $\sigma: I \rightarrow I$, $\Phi(\succsim_I^\sigma) = \Phi(\succsim_I)$.
\end{definition}

Since every individual's concern for $j$ is a positive multiple of $j$'s self-regarding utility, every individual has the same self-regarding preferences in $\succsim_I^\sigma$ as in $\succsim_I$. Anonymity thus requires that society disregard whether a given concern for $j$ is held by $j$ themself or by others.

We are ready for the characterization.

\begin{theorem} \label{theoremgeom}
    Suppose that $n \geq 3$.\footnote{This assumption cannot be dropped, and neither can the assumption in Section \ref{objective} that each $X_i$ contains at least two outcomes. See the remark after the proof in Section \ref{pf.geom} of the appendix.} Then $\Phi$ satisfies independence of outside concerns and is unanimous and anonymous if and only if it is geometric self-regarding utilitarian.
\end{theorem}

We provide a sketch of the proof. First, on this domain, independence of outside concerns and unanimity imply individual sovereignty. Since every individual is self-regarding utilitarian, in any profile every individual's comparisons in $\mathcal{S}_{\{i\}}$ agree with $i$'s self-regarding preferences, and so they agree with the comparisons of a profile in which all individuals have the same preferences, with the same self-regarding utilities. By unanimity, society has these preferences in the latter profile, and so by independence of outside concerns, society agrees with $i$'s self-regarding preferences on $\mathcal{S}_{\{i\}}$ in the former profile as well. This argument relies on the domain: if some individual could disagree with $i$ about $i$'s own outcome, independence of outside concerns and unanimity would not determine society's comparisons in $\mathcal{S}_{\{i\}}$.

Next, independence of outside concerns requires that, between profiles with the same self-regarding preferences, society's rate of substitution between any two individuals $j$ and $k$ be a function of only the rates of substitution between $j$ and $k$ held by the individuals. The fact that both society and every individual are self-regarding utilitarian implies that these aggregation functions (one for every pair of individuals) satisfy a particular consistency condition --- namely that for any third individual $\ell$, the aggregate of the individuals' rates of substitution between $j$ and $k$, times the aggregate of their rates between $k$ and $\ell$, must be the aggregate of their rates between $j$ and $\ell$. This is because for society and for each individual, the product of the rates of substitution between $j$ and $k$ and between $k$ and $\ell$ is the rate of substitution between $j$ and $\ell$. This condition, in conjunction with unanimity and anonymity, implies that each of these aggregation functions must be the geometric mean. The condition involves a third individual, which is where we use the assumption that $n \geq 3$.

This step of our proof is related to a characterization of the weighted geometric mean by \cite{aczel1986synthesis}, which requires that the synthesis of individuals' judgments of ratios be consistent with the synthesis of their judgments of measures. In our setting, this consistency is implied rather than assumed, since society's rates of substitution are ratios of the weights it places on self-regarding utilities.

The three axioms are independent, and each of the following social welfare functions satisfies individual sovereignty. Relative utilitarianism, under which $\Phi(\succsim_I)$ is represented by $\sum_{i \in I}\hat{u}_i$, where $\hat{u}_i$ is $i$'s utility index normalized so that its supremum is one and its infimum is zero (and $\Phi(\succsim_I)$ is complete indifference if every individual is completely indifferent), is unanimous and anonymous but violates independence of outside concerns. Normalizing an individual's utility makes society's comparison between $j$ and $k$ depend on how much that individual cares about others. For constants $\lambda_j > 0$ that are not all equal, the social welfare function under which $\Phi(\succsim_I)$ is represented by $\sum_{j \in I}\lambda_j\left(\prod_{i \in I}\alpha_{i,j}\right)^{\frac{1}{n}}u_{j,j}$ is anonymous and satisfies independence of outside concerns but is not unanimous. For constants $\gamma_i > 0$ that are not all equal, the social welfare function under which $\Phi(\succsim_I)$ is represented by $\sum_{j \in I}\left(\prod_{i \in I}\alpha_{i,j}^{\gamma_i}\right)^{1 / \sum_{i' \in I}\gamma_{i'}}u_{j,j}$, a weighted geometric mean as in \cite{aczel1986synthesis}, is unanimous and satisfies independence of outside concerns but is not anonymous.

Geometric self-regarding utilitarianism also responds naturally to changes in interpersonal comparisons: it is \emph{homogeneous in interpersonal comparisons}. Between profiles with the same self-regarding preferences, if every individual multiplies their rates of substitution between $j$ and $k$ by a common factor $b > 0$, then society's rates of substitution between $j$ and $k$ are also multiplied by $b$. For instance, suppose that each individual's outcome is the school they attend, that student $j$ prefers a school closer to home, and that student $k$ prefers a school with a particular academic program. If $j$'s family buys a car, and as a result every individual, including $j$, halves the importance of $j$ attending the closer school relative to $k$ attending the school with the program, then society halves it as well. Homogeneity strengthens independence of outside concerns for pairs of individuals between profiles with the same self-regarding preferences, which is the case $b = 1$.\footnote{Geometric self-regarding utilitarianism is also immune to gerrymandering --- it satisfies \citeauthor{chambers2008consistent}'s \citeyearpar{chambers2008consistent} \emph{representative consistency}. Suppose that individuals are partitioned into groups, that each group's preferences are aggregated by geometric self-regarding utilitarianism, with geometric means taken over the group, and that each individual's preferences are then replaced by their group's aggregate preferences. Aggregating the resulting profile yields the same social preferences as aggregating the original profile directly, regardless of the partition, since the geometric mean of the groups' geometric means, weighted by group size, is the geometric mean over all individuals.}

Our result is related to other multiplicative social welfare functions in the literature. \cite{demeyer1971welfare} examine a setting where the \emph{intensity} of preferences is represented by ratio-scale utility, i.e. utility that is unique up to multiplication by a positive constant. They characterize a multiplicative social welfare function that also yields a ratio-scale utility. \cite{kaneko1979nash}, in a setting similar to ours, examine the aggregation of vNM preferences. They assume there is a reference outcome that is universally the least-preferred outcome, and they characterize a multiplicative social welfare function with respect to the reference outcome.\footnote{\cite{sprumont2019nash} characterizes a similar social welfare function for which each individual's reference outcome is their worst outcome.} \cite{saaty2012possibility} examine the aggregation of pairwise comparisons between alternatives, and show that a two-step aggregation procedure --- which i) aggregates sets of pairwise comparisons into a set of social pairwise comparisons using the geometric mean, and ii) derives a cardinal ranking from the social pairwise comparisons --- satisfies Arrow's conditions adapted to the setting. These papers are similar in that they aggregate ratio comparisons of \emph{alternatives}, whereas we consider the aggregation of comparisons between self-regarding utilities.

Making interpersonal comparisons from other-regarding preferences does raise a serious concern --- society inherits any prejudices that individuals might have towards others. For instance, an individual could assign lower weight to particular groups of people based on their race, gender, or class. Or someone with many friends or a large family could be weighted more heavily than someone who lacks them.

Note that in settings with other-regarding preferences, this concern is not unique to geometric self-regarding utilitarianism. For instance, relative utilitarianism, which sums utility normalized to the unit interval, would also have this property, since each individual's utility reflects their other-regarding preferences.

Suppose we take a descriptive interpretation of Theorem \ref{theoremgeom} --- that geometric self-regarding utilitarianism is a metaphor for how existing political systems make interpersonal comparisons --- rather than a normative one. Then the concern about prejudice is a justification for the existence of parts of the government that are not as directly responsive to preference aggregation as the legislature, such as the judicial system or the constitution, which can protect groups of people from prejudice.

\section{Constrained Outcomes} \label{constrain}

So far we have assumed that the set of social outcomes is the full $n$-fold Cartesian product of individual outcomes. But in most settings, outcomes are constrained, and only a subset of this product set is feasible. For instance, in a school choice setting, it is not feasible to assign to any school more students than its capacity. Any mechanism with transfers must consider the feasibility of transfers. When society considers climate policy, outcomes are constrained by physical limitations.

In such settings, infeasible outcomes are not only impossible for society to choose, but may also be difficult for individuals to imagine and form preferences over \citep{fishburn1967conjoint}. Under what conditions do our results hold when individuals and society have preferences only over a subset of the Cartesian product?

Let $Y \subseteq X$ be the set of feasible outcomes. For each $i \in I$, denote by $Y_i = \{x_i \in X_i : x \in Y\}$ the set of feasible outcomes for individual $i$ and by $Y_{-i} = \{x_{-i} \in X_{-i} :  x \in Y\}$ the set of feasible profiles of outcomes for individuals other than $i$. Suppose that individual preferences $\succsim_i \subseteq \Delta(Y) \times \Delta(Y)$ are now over lotteries over \emph{feasible} outcomes and are vNM-rational.\footnote{We assume that all lotteries over feasible outcomes are feasible. There are settings where some lotteries over feasible outcomes are not feasible (particularly when the set of lotteries is constrained by physical limitations, as is the case in climate policy). But as long as individuals are able to imagine each ex-post outcome in $Y$, it seems reasonable to assume they can form preferences over lotteries over feasible outcomes.}

\cite{fishburn1967conjoint} shows that mutual independence is still necessary and sufficient for additivity of the utility index.\footnote{Note that $Y \subseteq Y_i \times Y_{-i}$, but the set of feasible outcomes need not be the Cartesian product of $Y_i$ and $Y_{-i}$, so Proposition \ref{prop1} does not apply.}

\begin{proposition} \label{prop3}
    \citep{fishburn1967conjoint} Suppose $\succsim_i$ is vNM-rational with utility index $u_i$. $\succsim_i$ is mutually independent of $Y_i$ and $Y_{-i}$ if and only if there exist subutility indices $u_{i,i}: Y_i \rightarrow \mathbb{R}$ and $u_{i,-i}: Y_{-i} \rightarrow \mathbb{R}$ such that $u_i(x) = u_{i,i}(x_i) + u_{i,-i}(x_{-i})$ for any $x \in Y$.
\end{proposition}

The uniqueness of the self-regarding and other-regarding subutility indices depends on the structure of $Y$. Informally, if $Y$ is sparse, the self-regarding utilities are non-unique --- different choices of the self-regarding utilities can rank $i$'s individual outcomes in different orders. For each individual $i$, define the binary relation $C_i \subseteq Y \times Y$ as follows: for any $x, y \in Y$, $xC_iy$ if there is a finite sequence $x = x^1, x^2, \dots, y$ of outcomes in $Y$ for which, between any pair of adjacent outcomes $x^k, x^{k+1}$, either $x^k_i = x^{k+1}_i$, or $x^k_{-i} = x^{k+1}_{-i}$. Say that any such sequence \emph{$C_i$-connects} $x$ and $y$.

Note that $C_i$ is reflexive, symmetric, and transitive, and so is an equivalence relation. Let $\mathcal{B}_i$ be a partition of $Y$ where each block $B \in \mathcal{B}_i$ is an equivalence class with respect to $C_i$. That is, each $x \in Y$ is a member of exactly one block $B \in \mathcal{B}_i$, and for each block $B \in \mathcal{B}_i$, $x, y \in B$ if and only if $x C_i y$. For a given $u_i$ and a given block $B \in \mathcal{B}_i$, subutility indices are unique up to simultaneous transformations $u_{i,i}'(x_i) = u_{i,i}(x_i) + c$ and $u_{i,-i}'(x_{-i}) = u_{i,-i}(x_{-i}) - c$ for any $c \in \mathbb{R}$. Notably, these transformations can be made in each block independently of other blocks. So a mutually independent $\succsim_i$ does induce self-regarding preferences over individual outcomes \emph{within} each block $B$, but remains silent regarding outcomes \emph{across} blocks. The same is true for $i$'s other-regarding preferences.

$Y$ is \emph{$i$-connected} if for any $x, y \in Y$, $xC_i y$, or, equivalently, if $\mathcal{B}_i$ consists of a single block. $Y$ is \emph{$I$-connected} if $Y$ is $i$-connected for all $i \in I$. If $Y$ is $i$-connected, then $\succsim_i$ induces complete self-regarding preferences, and $u_{i,i}$ is unique up to positive affine transformations, as in Section \ref{objective}. We discuss the implications of these uniqueness properties after presenting the results on aggregation.

In this constrained setting, social preferences $\succsim_0 \subseteq \Delta(Y) \times \Delta(Y)$ are also over lotteries over feasible outcomes and are vNM-rational. Say that $\succsim_0$ satisfies (individual) sovereignty if for any $i \in I$ and any $P, Q \in \Delta(Y)$ where $P_{-i} = Q_{-i}$, $P \succsim_0 Q$ if and only if $P \succsim_i Q$.

On a subset of the Cartesian product, individual sovereignty alone need not imply that social preferences are self-regarding utilitarian. The following proposition characterizes individual sovereignty for any $Y$.

\begin{proposition} \label{prop.constrained}
    Suppose that each $\succsim_i$ satisfies mutual independence of $Y_i$ and $Y_{-i}$. $\succsim_0$ satisfies individual sovereignty if and only if for each $i \in I$, there exist $b_i > 0$ and $u_{0,-i}: Y_{-i} \rightarrow \mathbb{R}$ such that $u_0(x) = b_iu_{i,i}(x_i) + u_{0,-i}(x_{-i})$ for any $x \in Y$.
\end{proposition}

That is, for each individual $i$, social preferences satisfy mutual independence of $Y_i$ and $Y_{-i}$, and society's subutility index over $i$'s outcomes is a positive multiple of $i$'s self-regarding utility. Note that the condition does not depend on which subutility indices we choose for individual $i$. Two choices differ by a constant within each block of $\mathcal{B}_i$, and since outcomes that share $x_{-i}$ lie in the same block, such constants can be absorbed into $u_{0,-i}$.

Unlike in Theorem \ref{theorem1}, these requirements need not combine into a social utility index that is additive across individuals. They do if we can move between any two feasible outcomes by changing one individual's outcome at a time. For any $x, y \in Y$, say that $xCy$ if there is a finite sequence $x = x^1, x^2, \dots, y$ of outcomes in $Y$ for which, between any pair of adjacent outcomes $x^k, x^{k+1}$, $x^k_{-j} = x^{k+1}_{-j}$ for some individual $j \in I$. That is, adjacent outcomes differ only in the outcome of a single individual. $Y$ is \emph{connected} if for any $x, y \in Y$, $xCy$.

\begin{theorem} \label{theorem3}
    Suppose that each $\succsim_i$ satisfies mutual independence of $Y_i$ and $Y_{-i}$, and that $Y$ is connected. $\succsim_0$ satisfies individual sovereignty if and only if $u_0$ is a positive linear combination of $\{u_{i,i}\}_{i \in I}$.
\end{theorem}

Since $X$ is connected, Theorem \ref{theorem1} is a special case of Theorem \ref{theorem3}.\footnote{When $n = 2$, individual sovereignty implies that $u_0$ is a positive linear combination of self-regarding subutility indices for any $Y$, for a suitable choice of these indices. See the remark at the end of Section \ref{pf.theorem3} in the appendix.} Note that if $Y$ is connected, then $Y$ is $I$-connected --- between adjacent outcomes that differ only in $j$'s outcome, $x^k_i = x^{k+1}_i$ if $j \neq i$, and $x^k_{-i} = x^{k+1}_{-i}$ if $j = i$. So in Theorem \ref{theorem3}, self-regarding utilities are unique up to positive affine transformations, as in Theorem \ref{theorem1}.

When $Y$ is not connected, individual sovereignty can be much weaker, for two reasons.

First, if $Y$ is not $I$-connected, then for some individual $i$, $\succsim_i$ does not pin down $i$'s self-regarding utility up to positive affine transformations --- different choices of $u_{i,i}$ can rank $i$'s outcomes in different blocks of $\mathcal{B}_i$ in different orders. Where $\succsim_i$ remains silent, society can have a say --- the particular self-regarding utility $u_{i,i}$ in society's utility representation in Proposition \ref{prop.constrained} can arrange the blocks in $\mathcal{B}_i$ in any order, including making them overlap.

At the extreme, there are constrained outcome spaces for which society's preferences over an individual $i$'s outcome have \emph{no relation} to $i$'s preferences. Say that $Y$ is \emph{$i$-disconnected} if for any $B \in \mathcal{B}_i$ and any $x, y \in B$, $x_i = y_i$, i.e. each block contains only one individual outcome for $i$. Equivalently, the outcomes of others determine $i$'s outcome --- no two outcomes in $Y$ differ only in $i$'s outcome. Then any function on $Y$ can be written as a function of $x_{-i}$ alone, and so the condition for $i$ in Proposition \ref{prop.constrained} holds for any $u_0$. $Y$ is \emph{totally disconnected} if $Y$ is $i$-disconnected for all $i \in I$, i.e. no two outcomes in $Y$ differ only in the outcome of a single individual. The following example in house allocation illustrates a totally disconnected $Y$:

\begin{example}
    House allocation. Let $n = 3$ and $Y_i = \{a,b,c\}$ for all $i$. Outcomes $x \in Y$ are such that $x_i \neq x_j$ for distinct $i,j \in I$, as listed below, sorted according to $\mathcal{B}_1$. A given utility index $u_1$ can be decomposed into $u_{1,1}$ and $u_{1,-1}$, or $u_{1,1}'$ and $u_{1,-1}'$. The table displays utility values.
\begin{center}
    \begin{tabular}{L||L||L|L||L|L}
        x = (x_1, x_2, x_3) & u_{1}(x) & u_{1,1}(x_1) & u_{1, -1}(x_{-1}) & u_{1,1}'(x_1) & u_{1, -1}'(x_{-1}) \\
        \hline \hline
        (a,b,c) & 2 & 2 & 0 & 0 & 2 \\
        (a,c,b) & 2 & 2 & 0 & 0 & 2 \\
        \hline
        (b,a,c) & 1 & 1 & 0 & 1 & 0 \\
        (b,c,a) & 1 & 1 & 0 & 1 & 0 \\
        \hline
        (c,a,b) & 0 & 0 & 0 & 2 & -2 \\
        (c,b,a) & 0 & 0 & 0 & 2 & -2
    \end{tabular}
\end{center}
\end{example}
    The outcome space in the example (and generally house allocation with $n$ individuals and $n$ houses) is totally disconnected --- for each individual $i$, each individual outcome lives in a different block in $\mathcal{B}_i$. For individual 1 to receive house $b$ instead of $a$, house $a$ must go to another individual. Individual 1's preferences do not determine whether the drop in their utility from 2 to 1 is due to their own change in outcome or to the change in outcome for individual 2 or 3. In fact, the houses of individuals 2 and 3 determine the house of individual 1, so there are no two distinct lotteries $P, Q \in \Delta(Y)$ for which $P_{-1} = Q_{-1}$, and similarly for individuals 2 and 3. Individual sovereignty places no restrictions on social preferences --- society's utility index need not even be additive across individuals.

    Second, even if $Y$ is $I$-connected, so that each individual's self-regarding utility is unique up to positive affine transformations, society's utility index need not be additive across individuals. $I$-connectedness is not enough for the conclusion of Theorem \ref{theorem3}:

    \begin{proposition} \label{prop.Iconnected}
        Suppose that each $\succsim_i$ satisfies mutual independence of $Y_i$ and $Y_{-i}$, and that $Y$ is $I$-connected but not connected. Then there exists a $\succsim_0$ that satisfies individual sovereignty but whose utility index $u_0$ is not a positive linear combination of $\{u_{i,i}\}_{i \in I}$.
    \end{proposition}

    Informally, if $Y$ is not connected, there is a group of outcomes that cannot be reached from the rest of $Y$ by changing one individual's outcome at a time.\footnote{For example, let $n = 3$ and let $Y$ consist of the six outcomes in the house allocation example together with $(a,a,a)$, $(a,a,b)$, and $(a,a,c)$. $Y$ is $I$-connected, but no other outcome in $Y$ differs from $(b,c,a)$ only in the outcome of a single individual.} Any two lotteries that individual sovereignty compares place the same probability on this group, so society can favor every outcome in the group without violating sovereignty.

    We informally discuss examples of other design settings to illustrate the definitions of connectedness and total disconnectedness. While house allocation with $n$ individuals and $n$ houses is totally disconnected, allocation with $n+1$ houses is connected --- from any allocation, we can reach any other by moving one individual at a time to the vacant house. In practice, house $n+1$ could denote not being assigned a house. By similar reasoning, even with $n$ individuals and $n$ houses, if the designer is not required to allocate every house, then the outcome space is connected. A school choice setting where a student's outcome is the school they attend is similar to house allocation. A school choice setting with $n$ individuals and $n$ seats is totally disconnected, and one with more seats is connected. One-sided matching, where an individual's outcome is their partner, is totally disconnected --- matching $i$ to $j$ affects both $i$'s and $j$'s outcomes, and $i$'s changes in utility may be due to their own change in outcome or $j$'s. For the same reason, two-sided, one-to-one matching is totally disconnected if individuals on each side have other-regarding preferences towards individuals on the other side. If other-regarding preferences are restricted to concern only individuals on the same side, then the outcome space is connected so long as individuals can remain unmatched.

\section{Subjective Uncertainty} \label{subjective}

We now turn to a Savage setting, where individuals have heterogeneous beliefs. Let $S$ be the set of states of the world, with a $\sigma$-algebra $\Sigma$ of events. We present our results assuming that the set of social outcomes is the full $n$-fold Cartesian product $X = X_1 \times \dots \times X_n$. Analogous results hold when outcomes are constrained as in the previous section. Additionally assume that each $X_i$ is endowed with a $\sigma$-algebra, and that $X$ and each $X_{-i}$ are endowed with the product $\sigma$-algebras, so that marginal distributions on $X_i$ and $X_{-i}$ are defined. The set $F = \{f \mid f: S \rightarrow X, \text{$f$ is $\Sigma$-measurable}\}$ is the set of \emph{acts}.

In this setting, each individual $i$ has preferences $\succsim_i \subseteq F \times F$ over acts. Similarly to the objective setting, we assume that individual preferences admit a subjective expected utility representation --- say that $\succsim_i$ is \emph{SEU-rational} if there exist a nonconstant, bounded, measurable utility index $u_i: X \rightarrow \mathbb{R}$ and a countably additive and nonatomic probability measure $\pi_i$ on $\Sigma$, such that for any act $f \in F$, $\int_S u_i(f(s))d\pi_i(s)$ represents $\succsim_i$.\footnote{\cite{villegas1964qualitative} provides conditions on preferences guaranteeing representation by a subjective countably additive probability measure and a bounded utility function.}

We also assume an independence condition on individuals' preferences. To state the condition, we introduce the following notation. For any individual $i \in I$ and act $f \in F$, let the \emph{induced} probability distribution of act $f$ according to individual $i$, $P^{f, i}$, be defined as
\begin{align*}
    P^{f, i}(Y)  = \pi_i\left(\{s \in S: f(s) \in Y\}\right)
\end{align*}
for any measurable $Y \subseteq X$. That is, individual $i$ believes that act $f$ would yield a probability distribution $P^{f, i}$ over outcomes. Denote by $P^{f, i}_j$ the marginal distribution over $X_{j}$ and by $P^{f, i}_{-j}$ the marginal distribution over $X_{-j}$.

\begin{definition}
    $\succsim_i$ over acts satisfies \emph{mutual independence of $X_i$ and $X_{-i}$} if for any acts $f,g \in F$ such that the induced marginal distributions satisfy $P^{f,i}_i = P^{g,i}_i$ and $P^{f, i}_{-i} = P^{g,i}_{-i}$, $f \sim_i g$.
\end{definition}

In this subjective uncertainty setting, mutual independence is also equivalent to the existence of an additive utility index.

\begin{proposition} \label{prop2}
    $\succsim_i$ satisfies mutual independence of $X_i$ and $X_{-i}$ if and only if there exist \emph{subutility indices }$u_{i,i}: X_i \rightarrow \mathbb{R}$ and $u_{i,-i}: X_{-i} \rightarrow \mathbb{R}$ such that $u_i(x) = u_{i,i}(x_i) + u_{i,-i}(x_{-i})$ for any $x \in X$. Furthermore, for a given $u_i$, subutility indices $u_{i,i}$ and $u_{i,-i}$ are unique up to simultaneous transformations $u_{i,i}' = u_{i,i} + c$ and $u_{i,-i}' = u_{i,-i} - c$ for any $c \in \mathbb{R}$.
\end{proposition}

Let social preferences $\succsim_0 \subseteq F \times F$ be over acts. As in \cite{gilboa2004utilitarian}, we assume that $\succsim_0$ is SEU-rational, with utility index $u_0$ and belief $\pi_0$.

What definition of individual sovereignty is appropriate in this subjective setting? Here is a candidate definition: $\succsim_0$ satisfies \emph{unrestricted (individual) sovereignty} if for any $i \in I$ and any $f, g \in F$ where $P^{f,i}_{-i} = P^{g,i}_{-i}$, $f \succsim_0 g$ if and only if $f \succsim_i g$.

But in this setting with subjective uncertainty, there may be multiple individuals $i$ and $j$ for whom $P^{f,i}_{-i} = P^{g,i}_{-i}$ and $P^{f,j}_{-j} = P^{g,j}_{-j}$, and their preferences over $f$ and $g$ may disagree. The following example illustrates this.

\begin{example}
    Let $I = \{1,2\}$ and $X_i = \{x_i, y_i\}$ for each $i =1,2$. Suppose that for each $i = 1,2$, $\succsim_i$ is SEU-rational and satisfies mutual independence of $X_i$ and $X_{-i}$, and that $u_{i,i}(x_i) > u_{i,i}(y_i)$. Consider a partition $E_1, E_2, E_3$ of $S$, for which $\pi_1(E_1) = 1/2$, $\pi_1(E_2) = \pi_1(E_3) = 1/4$, $\pi_2(E_1) = \pi_2(E_2) = 1/4$, and $\pi_2(E_3) = 1/2$. Consider acts $f, g \in F$, where $f(s) = (x_1, x_2)$ for $s \in E_1$, $f(s) = (x_1, y_2)$ for $s \in E_2$, $f(s) = (y_1, y_2)$ for $s \in E_3$, $g(s) = (y_1, y_2)$ for $s \in E_1$, $g(s) = (y_1, x_2)$ for $s \in E_2$, $g(s) = (x_1, x_2)$ for $s \in E_3$. Note that $P^{f,1}_2 = P^{g,1}_2$ and $P^{f,2}_1 = P^{g,2}_1$. But $f \succ_1 g$ and $f \prec_2 g$, and so social preferences cannot respect unrestricted sovereignty.
\end{example}

Unrestricted sovereignty goes awry because society defers not only to individuals' utility, but also to their beliefs about when society should defer. The example illustrates that multiple individuals can each believe that their outcome is the only one that differs between two acts, and thus that they each should dictate social preferences, potentially leading to a contradiction. For those who are sympathetic to unrestricted sovereignty as a normative desideratum, the example also demonstrates an impossibility result.

In the spirit of \cite{gilboa2004utilitarian}, we restrict sovereignty to hold only between pairs of acts for which all individuals agree on the induced probability distributions over outcomes.

First, define $\Lambda$ to be the set of events such that for each event $E \in \Lambda$, every individual is in agreement about the probability of $E$:
\begin{align*}
    \Lambda = \{E \in \Sigma: \text{for all $i,j \in I$, }\pi_i(E) = \pi_j(E)\}.
\end{align*}

Call an act $f$ a \emph{lottery} if for every measurable $Y \subseteq X$, $f^{-1}(Y) \in \Lambda$. Denote by $\hat{F}$ the set of lotteries. Note that for any lottery $f $, $P^{f,i}= P^{f, j}$ for any $i,j \in I$, so we abbreviate the notation to $P^f$.
\begin{definition}
    $\succsim_0$ satisfies \emph{restricted (individual) sovereignty} if for any $i \in I$ and any lotteries $f, g \in \hat{F}$ where $P^{f}_{-i} = P^{g}_{-i}$, $f \succsim_0 g$ if and only if $f \succsim_i g$.
\end{definition}

Since restricted sovereignty only applies to acts for which all beliefs are in agreement, society does not defer to potentially conflicting beliefs.

Before stating our main result in this setting, we introduce a mild condition on the preference profile.
\begin{definition}
    \emph{Minimal self-interest}. There exists an individual $i$ for whom there exist $x_i, y_i \in X_i$ and $x_{-i} \in X_{-i}$ such that $f \succ_i g$, where $f(s) = (x_i, x_{-i})$ and $g(s) = (y_i, x_{-i})$ for all $s \in S$.
\end{definition}

That is, there exists an individual $i$ who has a strict preference between two specific constant acts which yield different outcomes for $i$ but the same outcomes for other individuals.

We can now state our main result in this setting.

\begin{theorem} \label{theorem2}
    Suppose that each $\succsim_i$ satisfies mutual independence of $X_i$ and $X_{-i}$ and that minimal self-interest holds. $\succsim_0$ satisfies restricted individual sovereignty if and only if $u_0$ is a positive linear combination of $\{u_{i,i}\}_{i \in I}$ and $\pi_0$ is an affine combination of $\{\pi_i\}_{i \in I}$.
\end{theorem}

Say that such $\succsim_0$ is \emph{subjective self-regarding utilitarian}, since society's utility is again a weighted sum of individuals' self-regarding utilities. As before, other-regarding preferences are relegated to determining the weights on utilities and beliefs.

While restricted sovereignty does not require society to defer to individual beliefs as unrestricted sovereignty does, in conjunction with the other conditions, it does require society to defer to individuals on the probability of events on which individuals unanimously agree. \cite{mongin1995consistent} shows that this is sufficient for society's beliefs to be an affine combination of individual beliefs. The additive form of utility follows similarly to Theorem \ref{theorem1}. Here we use a result from \cite{dubins1961cut} to construct the desired lotteries in the proof.

\section{Conclusion} \label{conclusion}

In our setting, individuals have self-regarding and other-regarding preferences over a Cartesian product of sets of individual outcomes. We introduce a normative desideratum called individual sovereignty, and show that individual sovereignty implies that social preferences are self-regarding utilitarian, which relegates other-regarding preferences to just determining the weights placed on each individual's self-regarding utility. One particular weighting method --- geometric self-regarding utilitarianism --- hinges on independence of outside concerns: when only a group of individuals is at stake, society's preferences depend only on individuals' preferences towards that group. We provide extensions to a setting where feasible social outcomes are a subset of the Cartesian product and to a setting with subjective uncertainty.

\small 
\setlength{\bibsep}{0pt}
\bibliography{draft/bibliography}

\appendix

\section{Proofs} \label{proofs}

Throughout the appendix, for a function $u$ and a lottery $P$ over its domain, we write $u(P)$ for the expected value of $u$ under $P$.

\subsection{Proof of Theorem \ref{theorem1}}

\emph{Sufficiency of sovereignty.}
\begin{proof}

Fix $i \in I$, and consider lotteries $P, Q \in \Delta(X)$ such that $P_i = Q_i$ and $P_{-i} = Q_{-i}$. By mutual independence, $P \sim_i Q$. By sovereignty, $P \sim_0 Q$. Thus, $\succsim_0$ is mutually independent of $X_i$ and $X_{-i}$. Then, by Proposition \ref{prop1}, there exist $f_i: X_i \rightarrow \mathbb{R}$ and $g_i: X_{-i} \rightarrow \mathbb{R}$ such that $u_0(x) = f_i(x_i) + g_i(x_{-i})$ for any $x \in X$. Note that Proposition \ref{prop1} gives such a decomposition for some utility index of $\succsim_0$, and so also for $u_0$, which is a positive affine transformation of it.

We now show that $u_0$ is additively separable across all individuals. Let $I = \{1, \dots, n\}$, and pick any $x^* = (x_1^*, \dots, x_n^*) \in X$. For any $x \in X$ and $k = 0, \dots, n$, let $x^k = (x_1, \dots, x_k, x_{k+1}^*, \dots, x_n^*)$, so that $x^0 = x^*$ and $x^n = x$. Note that $x^k$ and $x^{k-1}$ differ only in their $k$th components, and so $u_0(x^k) - u_0(x^{k-1}) = f_k(x_k) - f_k(x_k^*)$. Summing over $k$,
\begin{align*}
    u_0(x) = u_0(x^*) + \sum_{k \in I}\left(f_k(x_k) - f_k(x_k^*)\right).
\end{align*}
Then there exist functions $\{u_{0,i}\}_{i \in I}$ such that $u_0(x) = \sum_{i \in I} u_{0,i}(x_i)$.\footnote{In particular, $\succsim_0$ satisfies mutual independence of $X_1, \dots, X_n$ in the sense of \cite{fishburn1965independence}, i.e. $P \sim_0 Q$ for any lotteries $P, Q \in \Delta(X)$ for which $P_i = Q_i$ for all $i \in I$. We build on \cite{fishburn1965independence} by showing that any preference that satisfies mutual independence of $X_i$ and $X_{-i}$ for all $i \in I$ also satisfies mutual independence of $X_1, \dots, X_n$. The argument relies on the product structure of $X$. When only a subset of $X$ is feasible, mutual independence for each individual need not imply additivity across individuals; see Section \ref{constrain}.}

For each $\succsim_i$, fix some representation $u_i = u_{i,i} + u_{i,-i}$. We want to show that each $u_{0,i}$ is a positive affine transformation of $u_{i,i}$. Fix $i \in I$ and $\bar{x}_{-i} \in X_{-i}$. For any $P_i, Q_i \in \Delta(X_i)$, let $P$ and $Q$ be the products of $P_i$ and $Q_i$, respectively, with the degenerate lottery on $\bar{x}_{-i}$. Note that $P_{-i} = Q_{-i}$. Then
    \begin{align*}
        u_{0,i}(P_i) \geq u_{0,i}(Q_i) \iff P \succsim_0 Q \iff P \succsim_i Q \iff u_{i,i}(P_i) \geq u_{i,i}(Q_i),
    \end{align*}
    where the second equivalence is by individual sovereignty. That is, $u_{0,i}$ is also a utility index for $\succsim_{i,i}$. If $\succsim_{i,i}$ is not complete indifference, then by the uniqueness of utility indices up to positive affine transformations, there exist $a_i > 0$ and $c_i \in \mathbb{R}$ such that $u_{0,i} = a_iu_{i,i} + c_i$. If $\succsim_{i,i}$ is complete indifference, then $u_{i,i}$ and $u_{0,i}$ are both constant, and the same holds with $a_i = 1$. Therefore, $u_0(x) = c + \sum_{i \in I} a_i u_{i,i}(x_i)$, where $c = \sum_{i \in I} c_i$. The $c$ can be dropped since $u_0$ is unique up to positive affine transformation.

\end{proof}

\emph{Necessity of sovereignty.}

\begin{proof}
Suppose that $u_0 = \sum_{j \in I}a_ju_{j,j}$, where each $a_j > 0$. Consider any $i \in I$ and $P, Q \in \Delta(X)$ for which $P_{-i} = Q_{-i}$. Since $P_j = Q_j$ for all $j \neq i$,
\begin{align*}
    u_0(P) - u_0(Q) = a_i\left(u_{i,i}(P_i) - u_{i,i}(Q_i)\right).
\end{align*}
Since $P_{-i} = Q_{-i}$, $u_i(P) - u_i(Q) = u_{i,i}(P_i) - u_{i,i}(Q_i)$. Since $a_i > 0$, $P \succsim_0 Q$ if and only if $P \succsim_i Q$.
\end{proof}

\subsection{Proof of Theorem \ref{theoremgeom}} \label{pf.geom}

We first define rates of interpersonal substitution formally. Consider any $\succsim_I \in \mathcal{D}$, and any $\succsim \in \mathcal{R}$ that is self-regarding utilitarian with respect to $\succsim_I$, e.g. $\succsim_i$ for some $i \in I$, or $\Phi(\succsim_I)$ when $\Phi$ satisfies individual sovereignty. For any distinct $j, k \in I$ and any $x_j, y_j \in X_j$ and $x_k, y_k \in X_k$ for which $x_j \succ_{j,j} y_j$ and $x_k \succ_{k,k} y_k$, let $\mu^{\succsim}_{j,k}((x_j, y_j), (x_k, y_k))$ be $\succsim$'s rate of interpersonal substitution of the difference between $x_j$ and $y_j$ and that between $x_k$ and $y_k$, defined as follows. Since $\succsim$ is represented by a positive linear combination of the self-regarding utilities of $\succsim_I$, $(x_j, x_k, x_{-jk}) \succ (x_j, y_k, x_{-jk}) \succ (y_j, y_k, x_{-jk})$ for any $x_{-jk}$. Since vNM preferences are continuous, there exists an $\alpha \in (0,1)$ such that the following indifference holds: $\alpha(x_j, x_k, x_{-jk}) + (1-\alpha)(y_j, y_k, x_{-jk}) \sim (x_j, y_k, x_{-jk})$. Similarly, $(x_j, x_k, x_{-jk}) \succ (y_j, x_k, x_{-jk}) \succ (y_j, y_k, x_{-jk})$ for any $x_{-jk}$, and there exists a $\beta \in (0,1)$ such that $\beta(x_j, x_k, x_{-jk}) + (1-\beta)(y_j, y_k, x_{-jk}) \sim (y_j, x_k, x_{-jk})$. Let $\mu^{\succsim}_{j,k}((x_j, y_j), (x_k, y_k)) = \frac{\alpha}{\beta}$. Each comparison in this definition is in $\mathcal{S}_{\{j,k\}}$, and by Lemma \ref{lemma.rates} below, the rate does not depend on $x_{-jk}$. For an individual $i$, we write $\mu^i_{j,k}((x_j, y_j), (x_k, y_k))$ for $\mu^{\succsim_i}_{j,k}((x_j, y_j), (x_k, y_k))$.

\begin{figure}
    \centering
    \begin{tikzpicture}[>=stealth]

  % --- Coordinate definitions ---
  % Vertical extent: y=0 (bottom) to y=3.6 (top)  [reduced from 6]
  % Left bar x-position: 1.4,  Right bar x-position: 3.2  [closer together]

  \def\ybot{0}
  \def\ytop{3.6}
  \def\xleft{1.4}    % x-position of left bar
  \def\xright{3.2}   % x-position of right bar
  \def\xcenter{2.3}  % horizontal center for top/bottom labels

  % Heights of the two middle nodes
  \def\yleftnode{2.4}   % 2/3 of 3.6
  \def\yrightnode{1.2}  % 1/3 of 3.6

  % --- Vertical bars ---
  % Left bar
  \draw[thick] (\xleft, \ybot) -- (\xleft, \ytop);
  % Right bar
  \draw[thick] (\xright, \ybot) -- (\xright, \ytop);

  % --- Tick marks at the node heights on each bar ---
  \draw[thick] (\xleft-0.1, \yleftnode) -- (\xleft+0.1, \yleftnode);
  \draw[thick] (\xright-0.1, \yrightnode) -- (\xright+0.1, \yrightnode);

  % --- Alpha brace/label: bottom section of left bar (0 to 2/3) ---
  \draw[<->, thick, gray]
    (\xleft-0.45, \ybot) -- (\xleft-0.45, \yleftnode)
    node[midway, left, black] {$\alpha$};

  % --- Beta brace/label: bottom section of right bar (0 to 1/3) ---
  \draw[<->, thick, gray]
    (\xright+0.45, \ybot) -- (\xright+0.45, \yrightnode)
    node[midway, right, black] {$\beta$};

  % --- Four coordinate labels ---

  % Top center
  \node[above] at (\xcenter, \ytop)
    {$(x_j,\, x_k,\, x_{-jk})$};

  % Bottom center
  \node[below] at (\xcenter, \ybot)
    {$(y_j,\, y_k,\, x_{-jk})$};

  % Middle left (2/3 height)
  \node[left=0.5cm] at (\xleft, \yleftnode)
    {$(x_j,\, y_k,\, x_{-jk})$};

  % Middle right (1/3 height)
  \node[right=0.5cm] at (\xright, \yrightnode)
    {$(y_j,\, x_k,\, x_{-jk})$};

\end{tikzpicture}
    \caption{Illustration of the rate of interpersonal substitution between pairs of outcomes.}
    \label{substitutionfig}
\end{figure}
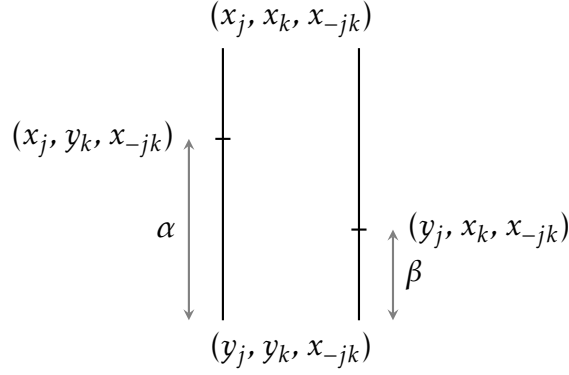

For example, suppose $\alpha = 2/3$ and $\beta = 1/3$. Then, for any utility index that represents $\succsim$, the utility difference between $(x_j, x_k, x_{-jk})$ and $(y_j, y_k, x_{-jk})$ is $3/2$ times that between $(x_j, y_k, x_{-jk})$ and $(y_j, y_k, x_{-jk})$. Similarly, the utility difference between $(x_j, x_k, x_{-jk})$ and $(y_j, y_k, x_{-jk})$ is three times that between $(y_j, x_k, x_{-jk})$ and $(y_j, y_k, x_{-jk})$. This implies that the utility difference between $(x_j, y_k, x_{-jk})$ and $(y_j, y_k, x_{-jk})$ is twice that between $(y_j, x_k, x_{-jk})$ and $(y_j, y_k, x_{-jk})$, which corresponds to $\mu_{j,k}^{\succsim}((x_j, y_j), (x_k, y_k)) = 2$. Figure \ref{substitutionfig} illustrates this particular example.

We next prove a lemma about rates of substitution.

\begin{lemma} \label{lemma.rates}
    Consider any distinct $j, k \in I$, $x_j, y_j \in X_j$, and $x_k, y_k \in X_k$. Suppose that $\succsim \in \mathcal{R}$ is represented by $u(x) = \sum_{\ell \in I} w_\ell u_{\ell,\ell}(x_\ell)$, where each $w_\ell > 0$, $u_{j,j}(x_j) > u_{j,j}(y_j)$, and $u_{k,k}(x_k) > u_{k,k}(y_k)$. Then for any $x_{-jk}$,
    \begin{align} \label{eq.rate}
        \mu^{\succsim}_{j,k}((x_j, y_j), (x_k, y_k)) = \frac{u(x_j, y_k, x_{-jk}) - u(y_j, y_k, x_{-jk})}{u(y_j, x_k, x_{-jk}) - u(y_j, y_k, x_{-jk})} = \frac{w_j\left(u_{j,j}(x_j) - u_{j,j}(y_j)\right)}{w_k\left(u_{k,k}(x_k) - u_{k,k}(y_k)\right)}.
    \end{align}
    The first equality also holds with $u$ replaced by any utility index that represents $\succsim$.
\end{lemma}

\begin{proof}
    In the definition of $\mu^{\succsim}_{j,k}((x_j, y_j), (x_k, y_k))$, $\alpha$ and $\beta$ satisfy
    \begin{align*}
        \alpha\left(u(x_j, x_k, x_{-jk}) - u(y_j, y_k, x_{-jk})\right) & = u(x_j, y_k, x_{-jk}) - u(y_j, y_k, x_{-jk}), \\
        \beta\left(u(x_j, x_k, x_{-jk}) - u(y_j, y_k, x_{-jk})\right) & = u(y_j, x_k, x_{-jk}) - u(y_j, y_k, x_{-jk}).
    \end{align*}
    Dividing the first equation by the second gives the first equality in (\ref{eq.rate}), and substituting $u(x) = \sum_{\ell \in I} w_\ell u_{\ell,\ell}(x_\ell)$ gives the second. The middle expression in (\ref{eq.rate}) does not change under positive affine transformations of $u$.
\end{proof}

We now prove the theorem.

\begin{proof}
    \emph{Individual sovereignty.} Suppose that $\Phi$ satisfies independence of outside concerns and is unanimous and anonymous. We first show that $\Phi$ satisfies individual sovereignty. Consider any $\succsim_I \in \mathcal{D}$ and $i \in I$, let $\{u_{\ell,\ell}\}_{\ell \in I}$ be self-regarding utilities of $\succsim_I$, and let $\succsim$ be the preference represented by $\sum_{\ell \in I}u_{\ell,\ell}$. Then $(\succsim, \dots, \succsim) \in \mathcal{D}$. For each $k \in I$, $\succsim_k$ and $\succsim$ agree on $\mathcal{S}_{\{i\}}$, since both compare the lotteries in each pair in $\mathcal{S}_{\{i\}}$ by the expected value of $u_{i,i}$. So by independence of outside concerns and unanimity, $\Phi(\succsim_I)$ agrees with $\Phi(\succsim, \dots, \succsim) = \succsim$ on $\mathcal{S}_{\{i\}}$, and hence with $\succsim_i$. So $\Phi$ satisfies individual sovereignty, and by Theorem \ref{theorem1}, $\Phi(\succsim_I)$ is self-regarding utilitarian with respect to $\succsim_I$ for every $\succsim_I \in \mathcal{D}$.

    \emph{Setup.} Consider any subutility indices $\{u_{\ell,\ell}\}_{\ell \in I}$, and let $\mathcal{D}_u \subseteq \mathcal{D}$ be the set of profiles whose self-regarding utilities are $\{u_{\ell,\ell}\}_{\ell \in I}$, i.e. the profiles in which each individual $\ell$'s self-regarding preferences are represented by $u_{\ell,\ell}$. All profiles in $\mathcal{D}_u$ have the same self-regarding preferences. Since every profile in $\mathcal{D}$ is in $\mathcal{D}_u$ for some $\{u_{\ell,\ell}\}_{\ell \in I}$, it is sufficient to show that $\Phi(\succsim_I)$ is geometric self-regarding utilitarian for every $\succsim_I \in \mathcal{D}_u$. If every $u_{\ell,\ell}$ is constant, then every individual is completely indifferent in each $\succsim_I \in \mathcal{D}_u$, so $\succsim_I$ is a profile in which all individuals have the same preferences, and by unanimity, $\Phi(\succsim_I)$ is complete indifference, which is geometric self-regarding utilitarian. So suppose that not every $u_{\ell,\ell}$ is constant. Steps 1 and 2 below apply to any $\{u_{\ell,\ell}\}_{\ell \in I}$, which we fix until Step 3.

    Say that an individual is \emph{self-concerned} if their self-regarding preference is not complete indifference. Let $I^* \subseteq I$ be the set of self-concerned individuals, i.e. the individuals $\ell$ for which $u_{\ell,\ell}$ is nonconstant; this set is nonempty. For each $\ell \in I^*$, fix individual outcomes $x_\ell, y_\ell \in X_\ell$ for which $u_{\ell,\ell}(x_\ell) > u_{\ell,\ell}(y_\ell)$. Since positive affine transformations of each $u_{\ell,\ell}$ represent the same self-regarding preferences, assume without loss of generality that $u_{\ell,\ell}(x_\ell) = 1$ and $u_{\ell,\ell}(y_\ell) = 0$ for each $\ell \in I^*$, and that $u_{\ell,\ell} = 0$ for each $\ell \in I \setminus I^*$. Then for any $\succsim_I \in \mathcal{D}_u$, we can express each $u_i$ and the social utility index $u_0$ of $\Phi(\succsim_I)$ as
    \begin{align*}
        u_i(x) = \sum_{\ell \in I^*}\alpha_{i,\ell}u_{\ell,\ell}(x_\ell), \qquad u_0(x) = \sum_{\ell \in I^*}\omega_\ell u_{\ell,\ell}(x_\ell),
    \end{align*}
    where each $\alpha_{i,\ell} > 0$ and $\omega_\ell > 0$. Conversely, for any $\alpha_{i,\ell} > 0$ for $i \in I$ and $\ell \in I^*$, the profile in which each $\succsim_i$ is represented by $u_i$ above is in $\mathcal{D}_u$.

    For distinct $j, k \in I^*$, abbreviate $\mu^i_{j,k} = \mu^i_{j,k}((x_j, y_j), (x_k, y_k))$ and $\mu^{\Phi}_{j,k} = \mu^{\Phi(\succsim_I)}_{j,k}((x_j, y_j), (x_k, y_k))$, and similarly $\mu^{i\prime}_{j,k}$ and $\mu^{\Phi\prime}_{j,k}$ for a profile $\succsim_I'$. By Lemma \ref{lemma.rates}, $\mu^i_{j,k} = \frac{\alpha_{i,j}}{\alpha_{i,k}}$ and $\mu^{\Phi}_{j,k} = \frac{\omega_j}{\omega_k}$. So it is sufficient to show that for every $\succsim_I \in \mathcal{D}_u$ and distinct $j, k \in I^*$,
    \begin{align} \label{eq.geomgoal}
        \mu^{\Phi}_{j,k} = \left(\prod_{i \in I}\mu^i_{j,k}\right)^{\frac{1}{n}}.
    \end{align}
    Indeed, (\ref{eq.geomgoal}) implies that $\frac{\omega_j}{\omega_k} = \left(\prod_{i \in I}\alpha_{i,j}\right)^{\frac{1}{n}} / \left(\prod_{i \in I}\alpha_{i,k}\right)^{\frac{1}{n}}$ for all $j, k \in I^*$, and so $u_0$ is a positive multiple of $\sum_{\ell \in I^*}\left(\prod_{i \in I}\alpha_{i,\ell}\right)^{\frac{1}{n}}u_{\ell,\ell}$. Since $u_{\ell,\ell} = 0$ for $\ell \notin I^*$, this is the utility index of geometric self-regarding utilitarianism. If $\vert I^* \vert = 1$, (\ref{eq.geomgoal}) holds vacuously, so suppose that $\vert I^* \vert \geq 2$.

    \emph{Step 1: social rates of substitution as functions of individual rates.} Fix distinct $j, k \in I^*$. Consider any $\succsim_I, \succsim_I' \in \mathcal{D}_u$ for which $\mu^i_{j,k} = \mu^{i\prime}_{j,k}$ for all $i \in I$, and let $\alpha_{i,\ell}'$ denote the weights in $\succsim_I'$. By Lemma \ref{lemma.rates}, $\frac{\alpha_{i,j}}{\alpha_{i,k}} = \frac{\alpha_{i,j}'}{\alpha_{i,k}'}$ for each $i \in I$. Consider any $(P, Q) \in \mathcal{S}_{\{j,k\}}$, so that $P_\ell = Q_\ell$ for each $\ell \neq j, k$. The difference between the expected utilities of $P$ and $Q$ given $u_i$ is
    \begin{align*}
        \alpha_{i,j}\left(u_{j,j}(P_j) - u_{j,j}(Q_j)\right) + \alpha_{i,k}\left(u_{k,k}(P_k) - u_{k,k}(Q_k)\right),
    \end{align*}
    and given $u_i'$ it is $\frac{\alpha_{i,k}'}{\alpha_{i,k}}$ times this difference. So $\succsim_i$ and $\succsim_i'$ agree on $\mathcal{S}_{\{j,k\}}$ for each $i \in I$, and by independence of outside concerns, $\Phi(\succsim_I)$ and $\Phi(\succsim_I')$ agree on $\mathcal{S}_{\{j,k\}}$. Since rates of substitution between $j$ and $k$ are defined by comparisons in $\mathcal{S}_{\{j,k\}}$, $\mu^{\Phi}_{j,k} = \mu^{\Phi\prime}_{j,k}$. Thus, for profiles in $\mathcal{D}_u$, $\mu^{\Phi}_{j,k}$ is a function of $(\mu^i_{j,k})_{i \in I}$. Let $\phi_{j,k}: \mathbb{R}^n_{>0} \rightarrow \mathbb{R}_{>0}$ be this function, i.e. $\phi_{j,k}((\mu^i_{j,k})_{i \in I}) = \mu^{\Phi}_{j,k}$ for any $\succsim_I \in \mathcal{D}_u$, where we list individuals in a fixed order.\footnote{$\phi_{j,k}$ depends on $\{u_{\ell,\ell}\}_{\ell \in I}$ and on $x_j, y_j, x_k, y_k$, but we drop this dependence for notational simplicity.} For any permutation $\sigma$ of $I$, $\succsim_I^\sigma \in \mathcal{D}_u$, and the individuals' rates of substitution in $\succsim_I^\sigma$ are those in $\succsim_I$, permuted. So by anonymity, $\phi_{j,k}$ is symmetric. Every $a \in \mathbb{R}^n_{>0}$ arises as $(\mu^i_{j,k})_{i \in I}$ for some $\succsim_I \in \mathcal{D}_u$, e.g. with $\alpha_{i,j} = a_i$ and $\alpha_{i,\ell} = 1$ for $\ell \in I^* \setminus \{j\}$. For $a, b \in \mathbb{R}^n_{>0}$, write $ab = (a_1b_1, \dots, a_nb_n)$, and let $\mathbf{1} = (1, \dots, 1)$.

    Since $\Phi$ is unanimous, $\phi_{j,k}$ is \emph{idempotent}: $\phi_{j,k}(c\mathbf{1}) = c$ for any $c > 0$. To see this, consider the preference $\succsim$ represented by $c\,u_{j,j} + \sum_{\ell \in I^* \setminus \{j\}}u_{\ell,\ell}$, for which $\mu^{\succsim}_{j,k} = c$. By unanimity, $\Phi(\succsim, \dots, \succsim) = \succsim$.

    \emph{Step 2: (\ref{eq.geomgoal}) when $\vert I^* \vert \geq 3$.} Suppose that $\vert I^* \vert \geq 3$, fix distinct $j, k \in I^*$, and let $\ell \in I^* \setminus \{j,k\}$. For each individual $i$, $\mu^i_{j,k}\mu^i_{k,\ell} = \mu^i_{j,\ell}$, and similarly for society. For any $a, b \in \mathbb{R}^n_{>0}$, there is $\succsim_I \in \mathcal{D}_u$ for which $(\mu^i_{j,k})_{i \in I} = a$ and $(\mu^i_{k,\ell})_{i \in I} = b$, e.g. with $\alpha_{i,j} = a_ib_i$, $\alpha_{i,k} = b_i$, and all other weights equal to one. Therefore,
    \begin{align*}
        \phi_{j,k}(a)\phi_{k,\ell}(b) = \phi_{j,\ell}(ab) \qquad \text{for all } a, b \in \mathbb{R}^n_{>0}.
    \end{align*}
    With $a = \frac{y}{x}$ and $b = x$ (componentwise), this is the quotient equation $\phi_{j,k}(\frac{y}{x}) = \frac{\phi_{j,\ell}(y)}{\phi_{k,\ell}(x)}$ of \cite{aczel1986synthesis}, who show that under mild regularity and idempotence (their unanimity property), its solutions are weighted geometric means. We give a short argument that uses anonymity in place of regularity.

    Setting $b = \mathbf{1}$ and using idempotence gives $\phi_{j,k} = \phi_{j,\ell}$, and setting $a = \mathbf{1}$ gives $\phi_{k,\ell} = \phi_{j,\ell}$. So $\phi_{j,k}(a)\phi_{j,k}(b) = \phi_{j,k}(ab)$ for all $a, b \in \mathbb{R}^n_{>0}$.

    For $c > 0$ and $i \in I$, let $(c, \mathbf{1}_{-i}) \in \mathbb{R}^n_{>0}$ be the vector whose $i$th component is $c$ and whose other components are one, and let $g(c) = \phi_{j,k}(c, \mathbf{1}_{-i})$, which does not depend on $i$ since $\phi_{j,k}$ is symmetric. Since $a$ is the componentwise product of the vectors $(a_i, \mathbf{1}_{-i})$ for $i \in I$,
    \begin{align*}
        \phi_{j,k}(a) = \prod_{i \in I}\,\phi_{j,k}(a_i, \mathbf{1}_{-i}) = \prod_{i \in I}\,g(a_i).
    \end{align*}
    In particular, by idempotence, $c = \phi_{j,k}(c\mathbf{1}) = g(c)^n$ for all $c > 0$, and since $g(c) > 0$, $g(c) = c^{\frac{1}{n}}$. So $\phi_{j,k}(a) = \left(\prod_{i \in I}a_i\right)^{\frac{1}{n}}$ for all $a \in \mathbb{R}^n_{>0}$, which is (\ref{eq.geomgoal}).

    \emph{Step 3: (\ref{eq.geomgoal}) when $\vert I^* \vert = 2$.} Let $I^* = \{j,k\}$, and consider any $\succsim_I \in \mathcal{D}_u$. Since $n \geq 3$, there is an individual $\ell \in I \setminus \{j,k\}$, and since $X_\ell$ contains at least two outcomes, there is a nonconstant $v_{\ell,\ell}: X_\ell \rightarrow \mathbb{R}$. Let $\{v_{\ell',\ell'}\}_{\ell' \in I}$ be the subutility indices for which $v_{\ell',\ell'} = u_{\ell',\ell'}$ for each $\ell' \neq \ell$, define $\mathcal{D}_v$ analogously to $\mathcal{D}_u$, and let $\succsim_I'$ be the profile in which each $\succsim_i'$ is represented by $u_i'(x) = u_i(x) + v_{\ell,\ell}(x_\ell)$. Then $\succsim_I' \in \mathcal{D}_v$, where the three individuals $j$, $k$, and $\ell$ are self-concerned, and each individual's rate of substitution between $(x_j, y_j)$ and $(x_k, y_k)$ is the same in $\succsim_I'$ as in $\succsim_I$. For any $(P, Q) \in \mathcal{S}_{\{j,k\}}$, $P_\ell = Q_\ell$, and so adding $v_{\ell,\ell}$ to $u_i$ changes the expected utilities of $P$ and $Q$ by the same amount. So $\succsim_i'$ and $\succsim_i$ agree on $\mathcal{S}_{\{j,k\}}$ for each $i \in I$, and by independence of outside concerns, $\Phi(\succsim_I')$ and $\Phi(\succsim_I)$ agree on $\mathcal{S}_{\{j,k\}}$. Since rates of substitution between $j$ and $k$ are defined by comparisons in $\mathcal{S}_{\{j,k\}}$, Step 2 applied to $\mathcal{D}_v$ gives
    \begin{align*}
        \mu^{\Phi}_{j,k} = \mu^{\Phi(\succsim_I')}_{j,k}((x_j, y_j), (x_k, y_k)) = \left(\prod_{i \in I}\mu^i_{j,k}\right)^{\frac{1}{n}}.
    \end{align*}
\end{proof}

We omit the proof that geometric self-regarding utilitarianism satisfies the axioms.

\begin{remark}
    Neither the assumption $n \geq 3$ of Theorem \ref{theoremgeom} nor the assumption that each $X_i$ contains at least two outcomes can be dropped. Specifically, the proof above requires at least three individuals, each with at least two outcomes. When $n = 2$, independence of outside concerns adds nothing beyond individual sovereignty. For $J = \{i\}$, every individual compares lotteries in $\mathcal{S}_{\{i\}}$ by $\succsim_{i,i}$, so profiles that agree on $\mathcal{S}_{\{i\}}$ have the same $\succsim_{i,i}$, and any $\Phi$ that satisfies individual sovereignty assigns them social preferences that agree on $\mathcal{S}_{\{i\}}$. For $J = I$, profiles that agree on $\mathcal{S}_I$ are identical. The axioms then require only that society's rate of substitution be a symmetric, idempotent function of the two individuals' rates.

    For example, suppose that $n = 2$ and that $X_1$ and $X_2$ each contain two outcomes, and fix $p \neq 0$. For each $\succsim_I \in \mathcal{D}$ in which both individuals are self-concerned, let $\Phi(\succsim_I)$ be the preference that is self-regarding utilitarian with respect to $\succsim_I$ and for which
    \begin{align*}
        \mu^{\Phi(\succsim_I)}_{1,2} = \left(\frac{1}{2}\sum_{i \in I}\left(\mu^i_{1,2}\right)^p\right)^{\frac{1}{p}},
    \end{align*}
    where the rates of substitution are between the two outcomes of individual 1 and the two outcomes of individual 2. That is, society's rate of substitution is the power mean, with exponent $p$, of the two individuals' rates. For any other $\succsim_I \in \mathcal{D}$, exactly one preference is self-regarding utilitarian with respect to $\succsim_I$; let $\Phi(\succsim_I)$ be that preference.

    Since each $\Phi(\succsim_I)$ is self-regarding utilitarian with respect to $\succsim_I$, $\Phi$ satisfies individual sovereignty, and so independence of outside concerns. It is also unanimous, since the power mean is idempotent, and anonymous, since the power mean is symmetric. But it is not geometric self-regarding utilitarian. The same is true if there is also a third individual whose set of individual outcomes contains only one outcome, with the power mean taken over all three individuals; then $\mathcal{S}_{\{1,2\}}$ contains every pair of lotteries.
\end{remark}
\subsection{Proof of Proposition \ref{prop.constrained}} \label{pf.constrained}

\emph{Sovereignty is sufficient.}

\begin{proof}
Fix an individual $i \in I$ and subutility indices $u_{i,i}$ and $u_{i,-i}$ such that $u_i(x) = u_{i,i}(x_i) + u_{i,-i}(x_{-i})$ for any $x \in Y$.

Consider lotteries $P, Q \in \Delta(Y)$ such that $P_i = Q_i$ and $P_{-i} = Q_{-i}$. By mutual independence, $P \sim_i Q$. By sovereignty, $P \sim_0 Q$. Thus, $\succsim_0$ is mutually independent of $Y_i$ and $Y_{-i}$. Then, by Proposition \ref{prop3}, there exist subutility indices $u_{0,i}: Y_i \rightarrow \mathbb{R}$ and $u_{0,-i}: Y_{-i} \rightarrow \mathbb{R}$ such that $u_0(x) = u_{0,i}(x_i) + u_{0,-i}(x_{-i})$ for any $x \in Y$.

We want to show that $u_{0,i}$ is a positive affine transformation of $u_{i,i}$ within each block $B \in \mathcal{B}_i$, with a positive scalar that is common across blocks. Note that we cannot proceed as in the proof of Theorem \ref{theorem1}, by fixing $\bar{x}_{-i}$ and varying $i$'s outcome over all of $Y_i$, since $(x_i, \bar{x}_{-i})$ need not be in $Y$.

Consider any $x_i, y_i, x_{-i}$ for which $(x_i, x_{-i}), (y_i, x_{-i}) \in Y$. Then $u_0(x_i, x_{-i}) - u_0(y_i, x_{-i}) = u_{0,i}(x_i) - u_{0,i}(y_i)$ and $u_i(x_i, x_{-i}) - u_i(y_i, x_{-i}) = u_{i,i}(x_i) - u_{i,i}(y_i)$. So by sovereignty over degenerate lotteries $(x_i, x_{-i})$ and $(y_i, x_{-i})$,
\begin{align} \label{eq.degenerate}
    u_{0,i}(x_i) \geq u_{0,i}(y_i) \iff u_{i,i}(x_i) \geq u_{i,i}(y_i).
\end{align}

\begin{claim}
    Consider any $x_i, y_i, x_{-i}$ and $x_i', y_i', x_{-i}'$ for which $(x_i, x_{-i}), (y_i, x_{-i}), (x_i', x_{-i}'), (y_i', x_{-i}') \in Y$, $u_{i,i}(x_i) > u_{i,i}(y_i)$, and $u_{i,i}(x_i') > u_{i,i}(y_i')$. Then
    \begin{align*}
        \frac{u_{0,i}(x_i) - u_{0,i}(y_i)}{u_{i,i}(x_i) - u_{i,i}(y_i)} = \frac{u_{0,i}(x_i') - u_{0,i}(y_i')}{u_{i,i}(x_i') - u_{i,i}(y_i')}.
    \end{align*}
\end{claim}
\begin{proof}
    Assume towards a contradiction that the left-hand side is greater than the right-hand side. Let
    \begin{align*}
        \alpha = \frac{u_{i,i}(x_i') - u_{i,i}(y_i')}{(u_{i,i}(x_i) - u_{i,i}(y_i)) + (u_{i,i}(x_i') - u_{i,i}(y_i'))},
    \end{align*}
    so that $\alpha(u_{i,i}(x_i) - u_{i,i}(y_i)) = (1-\alpha)(u_{i,i}(x_i') - u_{i,i}(y_i')) > 0$. Consider lotteries
    \begin{align*}
        P = \left(\alpha(x_i, x_{-i}), (1-\alpha)(y_i', x_{-i}')\right), \qquad Q = \left(\alpha(y_i, x_{-i}), (1-\alpha)(x_i', x_{-i}')\right).
    \end{align*}
    Note that $P_{-i} = Q_{-i}$. Since $\succsim_i$ is vNM-rational, $P \sim_i Q$ if and only if $\alpha(u_{i,i}(x_i) - u_{i,i}(y_i)) = (1-\alpha)(u_{i,i}(x_i') - u_{i,i}(y_i'))$, which holds by the choice of $\alpha$. Similarly, since $\succsim_0$ is vNM-rational and $u_0(x) = u_{0,i}(x_i) + u_{0,-i}(x_{-i})$, $P \succ_0 Q$ if and only if $\alpha(u_{0,i}(x_i) - u_{0,i}(y_i)) > (1-\alpha)(u_{0,i}(x_i') - u_{0,i}(y_i'))$. Dividing both sides by $\alpha(u_{i,i}(x_i) - u_{i,i}(y_i)) = (1-\alpha)(u_{i,i}(x_i') - u_{i,i}(y_i')) > 0$ yields our assumption. So $P \succ_0 Q$ and $P \sim_i Q$, which contradicts sovereignty.
\end{proof}

If $u_{i,i}(x_i) = u_{i,i}(y_i)$ for all $(x_i, x_{-i}), (y_i, x_{-i}) \in Y$, let $b_i = 1$. Otherwise, let $b_i$ be the common ratio in the claim, which is positive by (\ref{eq.degenerate}). Then for any $(x_i, x_{-i}), (y_i, x_{-i}) \in Y$,
\begin{align} \label{eq.ratio}
    u_{0,i}(x_i) - u_{0,i}(y_i) = b_i\left(u_{i,i}(x_i) - u_{i,i}(y_i)\right).
\end{align}
This follows from the claim if $u_{i,i}(x_i) > u_{i,i}(y_i)$, from the claim with $x_i$ and $y_i$ switched if $u_{i,i}(x_i) < u_{i,i}(y_i)$, and from (\ref{eq.degenerate}) if $u_{i,i}(x_i) = u_{i,i}(y_i)$.

Now fix a block $B \in \mathcal{B}_i$, and consider any $x, y \in B$ and a sequence that $C_i$-connects $x$ and $y$. Between any pair of adjacent outcomes $x^k, x^{k+1}$, either $x^k_i = x^{k+1}_i$, or $x^k_{-i} = x^{k+1}_{-i}$, and in both cases $u_{0,i}(x^k_i) - b_iu_{i,i}(x^k_i) = u_{0,i}(x^{k+1}_i) - b_iu_{i,i}(x^{k+1}_i)$ by (\ref{eq.ratio}). So there exists $a_B \in \mathbb{R}$ such that $u_{0,i}(x_i) = a_B + b_iu_{i,i}(x_i)$ for any $x \in B$. That is, within each block $B \in \mathcal{B}_i$, $u_{0,i}$ is a positive affine transformation of $u_{i,i}$, with the common scalar $b_i$.

Since outcomes that share $x_{-i}$ lie in the same block, $a_B$ is determined by $x_{-i}$. Let $u_{0,-i}'(x_{-i}) = u_{0,-i}(x_{-i}) + a_B$ for any $x \in B$ and any $B \in \mathcal{B}_i$. Then $u_0(x) = b_iu_{i,i}(x_i) + u_{0,-i}'(x_{-i})$ for any $x \in Y$. Since the fixed $i$ was arbitrary, this holds for all $i \in I$.
\end{proof}

\emph{Sovereignty is necessary.}

\begin{proof}
Consider any $i \in I$ and $P, Q \in \Delta(Y)$ for which $P_{-i} = Q_{-i}$. Then
\begin{align*}
    P \succsim_0 Q & \iff \int_Y \left(b_iu_{i,i}(x_i) + u_{0,-i}(x_{-i})\right)dP(x) \geq \int_Y \left(b_iu_{i,i}(x_i) + u_{0,-i}(x_{-i})\right)dQ(x) \\
    & \iff \int_{Y_i} u_{i,i}(x_i)dP_i(x_i) \geq \int_{Y_i} u_{i,i}(x_i)dQ_i(x_i) \\
    & \iff \int_Y \left(u_{i,i}(x_i) + u_{i,-i}(x_{-i})\right)dP(x) \geq \int_Y \left(u_{i,i}(x_i) + u_{i,-i}(x_{-i})\right)dQ(x) \\
    & \iff P \succsim_i Q,
\end{align*}
where the second and third equivalences hold since $P_{-i} = Q_{-i}$ and $b_i > 0$.
\end{proof}

\subsection{Proof of Theorem \ref{theorem3}} \label{pf.theorem3}

\emph{Sovereignty is sufficient.}

\begin{proof}
By Proposition \ref{prop.constrained}, for each $i \in I$ there exist $b_i > 0$ and $u_{0,-i}$ such that $u_0(x) = b_iu_{i,i}(x_i) + u_{0,-i}(x_{-i})$ for any $x \in Y$. So for any $x, y \in Y$ that differ only in individual $i$'s outcome, i.e. $x_{-i} = y_{-i}$,
\begin{align} \label{eq.step}
    u_0(x) - u_0(y) = b_i\left(u_{i,i}(x_i) - u_{i,i}(y_i)\right).
\end{align}

Pick any $x^* \in Y$, and consider any $x \in Y$. Since $Y$ is connected, there is a sequence $x^* = x^1, x^2, \dots, x^m = x$ of outcomes in $Y$ for which, for each $k$, $x^k$ and $x^{k+1}$ differ only in the outcome of some individual $j(k)$, i.e. $x^k_{-j(k)} = x^{k+1}_{-j(k)}$. By (\ref{eq.step}),
\begin{align*}
    u_0(x) & = u_0(x^*) + \sum_{k = 1}^{m-1}\left(u_0(x^{k+1}) - u_0(x^k)\right) \\
    & = u_0(x^*) + \sum_{k = 1}^{m-1} b_{j(k)}\left(u_{j(k),j(k)}(x^{k+1}_{j(k)}) - u_{j(k),j(k)}(x^k_{j(k)})\right) \\
    & = u_0(x^*) + \sum_{i \in I} b_i\sum_{k = 1}^{m-1}\left(u_{i,i}(x^{k+1}_i) - u_{i,i}(x^k_i)\right) \\
    & = u_0(x^*) + \sum_{i \in I} b_i\left(u_{i,i}(x_i) - u_{i,i}(x^*_i)\right),
\end{align*}
where the third equality holds since $x^{k+1}_i = x^k_i$ whenever $j(k) \neq i$, and the last equality holds since the inner sums telescope. Therefore $u_0(x) = c + \sum_{i \in I} b_iu_{i,i}(x_i)$ for any $x \in Y$, where $c = u_0(x^*) - \sum_{i \in I}b_iu_{i,i}(x^*_i)$. The $c$ can be dropped since $u_0$ is unique up to positive affine transformation.
\end{proof}

\emph{Sovereignty is necessary.} Suppose that $u_0(x) = \sum_{j \in I} b_ju_{j,j}(x_j)$ for any $x \in Y$, where each $b_j > 0$. Then for each $i \in I$, $u_0(x) = b_iu_{i,i}(x_i) + u_{0,-i}(x_{-i})$, where $u_{0,-i}(x_{-i}) = \sum_{j \neq i}b_ju_{j,j}(x_j)$. By Proposition \ref{prop.constrained}, $\succsim_0$ satisfies individual sovereignty.

\begin{remark}
When $n = 2$, Proposition \ref{prop.constrained} implies that $u_0$ is a positive linear combination of self-regarding subutility indices for any $Y$, for a suitable choice of these indices. By Proposition \ref{prop.constrained}, $u_0(x) = b_1u_{1,1}(x_1) + u_{0,-1}(x_2) = b_2u_{2,2}(x_2) + u_{0,-2}(x_1)$ for any $x \in Y$, and so
\begin{align*}
    u_{0,-1}(x_2) - b_2u_{2,2}(x_2) = u_{0,-2}(x_1) - b_1u_{1,1}(x_1).
\end{align*}
The left-hand side is the same for outcomes that share $x_2$, and the right-hand side is the same for outcomes that share $x_1$. So both sides are constant within each block $B \in \mathcal{B}_1$, say equal to $c_B$. (When $n = 2$, $\mathcal{B}_1 = \mathcal{B}_2$.) Let $u_{2,2}'(x_2) = u_{2,2}(x_2) + c_B/b_2$ for any $x \in B$ and any $B \in \mathcal{B}_2$, which is also a self-regarding subutility index for individual 2. Then $u_0(x) = b_1u_{1,1}(x_1) + b_2u_{2,2}'(x_2)$ for any $x \in Y$.
\end{remark}

\subsection{Proof of Proposition \ref{prop.Iconnected}}

\begin{proof}
Note that $C$ is reflexive, symmetric, and transitive, and so is an equivalence relation. Let $\mathcal{B}$ be a partition of $Y$ where each block is an equivalence class with respect to $C$. Since $Y$ is not connected, $\mathcal{B}$ has at least two blocks. Pick any block $B \in \mathcal{B}$. For each $j \in I$, fix subutility indices $u_{j,j}$ and $u_{j,-j}$ such that $u_j(x) = u_{j,j}(x_j) + u_{j,-j}(x_{-j})$ for any $x \in Y$, and let $\succsim_0$ have utility index
\begin{align*}
    u_0(x) = \begin{cases}
        1 + \sum_{j \in I} u_{j,j}(x_j) & \text{if } x \in B, \\
        \sum_{j \in I} u_{j,j}(x_j) & \text{if } x \notin B.
    \end{cases}
\end{align*}

We first show that $\succsim_0$ satisfies individual sovereignty. Consider any $i \in I$, and let $B_{-i} = \{x_{-i} \in Y_{-i}: x \in B\}$. If $x, y \in Y$ and $x_{-i} = y_{-i}$, then $xCy$, and so $x \in B$ if and only if $y \in B$. Thus for any $x \in Y$, $x \in B$ if and only if $x_{-i} \in B_{-i}$. Consider any $P, Q \in \Delta(Y)$ for which $P_{-i} = Q_{-i}$. Then $P(B) = P_{-i}(B_{-i}) = Q_{-i}(B_{-i}) = Q(B)$, and $P_j = Q_j$ for all $j \neq i$. So
\begin{align*}
    \int_Y u_0(x)dP(x) - \int_Y u_0(x)dQ(x) & = P(B) - Q(B) + \sum_{j \in I}\left(\int_{Y_j}u_{j,j}(x_j)dP_j(x_j) - \int_{Y_j}u_{j,j}(x_j)dQ_j(x_j)\right) \\
    & = \int_{Y_i}u_{i,i}(x_i)dP_i(x_i) - \int_{Y_i}u_{i,i}(x_i)dQ_i(x_i) \\
    & = \int_Y u_i(x)dP(x) - \int_Y u_i(x)dQ(x),
\end{align*}
where the last equality holds since $P_{-i} = Q_{-i}$. So $P \succsim_0 Q$ if and only if $P \succsim_i Q$.

We now show that $u_0$ is not a positive linear combination of $\{u_{j,j}\}_{j \in I}$. Assume towards a contradiction that it is. Since $Y$ is $I$-connected, self-regarding subutility indices are unique up to positive affine transformations, and so there exist $c \in \mathbb{R}$ and $w_j > 0$ for each $j \in I$ such that $u_0(x) = c + \sum_{j \in I}w_ju_{j,j}(x_j)$ for any $x \in Y$. Then for any $x \in Y$,
\begin{align*}
    c + \sum_{j \in I}(w_j - 1)u_{j,j}(x_j) = \begin{cases}
        1 & \text{if } x \in B, \\
        0 & \text{if } x \notin B.
    \end{cases}
\end{align*}
Fix any $i \in I$. For any $(x_i, x_{-i}), (y_i, x_{-i}) \in Y$, both outcomes are in $B$ or both are not, and so $(w_i - 1)u_{i,i}(x_i) = (w_i - 1)u_{i,i}(y_i)$. Now consider any $x, y \in Y$, and a sequence that $C_i$-connects $x$ and $y$, which exists since $Y$ is $i$-connected. Between any pair of adjacent outcomes $x^k, x^{k+1}$, either $x^k_i = x^{k+1}_i$, or $x^k_{-i} = x^{k+1}_{-i}$, and in both cases $(w_i - 1)u_{i,i}(x^k_i) = (w_i - 1)u_{i,i}(x^{k+1}_i)$. So $(w_i - 1)u_{i,i}(x_i) = (w_i - 1)u_{i,i}(y_i)$, i.e. $(w_i - 1)u_{i,i}$ is constant on $Y_i$. Since the fixed $i$ was arbitrary, the left-hand side of the display above is constant on $Y$. But it equals 1 on $B$ and 0 on any other block of $\mathcal{B}$, a contradiction.
\end{proof}

\subsection{Proof of Proposition \ref{prop2}}

\emph{Sufficiency of mutual independence of $X_i$ and $X_{-i}$.}

\begin{proof}

This proof follows the steps of \cite{fishburn1965independence}, and uses the nonatomicity of beliefs to construct the needed acts.

Pick any $x^* = (x_i^*, x_{-i}^*) \in X$, and let $u_{i,i}(x_i) = u_i(x_i, x_{-i}^*) - u_i(x_i^*, x_{-i}^*)$ for any $x_i \in X_i$ and $u_{i,-i}(x_{-i}) = u_i(x_i^*, x_{-i})$ for any $x_{-i} \in X_{-i}$. Then, for any $x_i \in X_i$ and $x_{-i} \in X_{-i}$,
    \begin{align*}
        u_{i,i}(x_i) + u_{i,-i}(x_{-i}) = u_i(x_i, x_{-i}^*) + u_i(x_i^*, x_{-i}) - u_i(x_i^*, x_{-i}^*),
    \end{align*}
    and so it is sufficient to show that the right-hand side equals $u_i(x_i, x_{-i})$.

    Since $\pi_i$ is nonatomic, there exist acts $f, g \in F$ such that $P^{f,i} = \left(\frac{1}{2}(x_i, x_{-i}), \frac{1}{2}(x_i^*, x_{-i}^*)\right)$ and $P^{g,i} = \left(\frac{1}{2}(x_i^*, x_{-i}), \frac{1}{2}(x_i, x_{-i}^*)\right)$. Since $\succsim_i$ satisfies mutual independence of $X_i$ and $X_{-i}$, $f \sim_i g$. Since $\succsim_i$ is SEU-rational,
    \begin{align*}
        & \int_S u_i(f(s)) d \pi_i(s) = \frac{1}{2}u_i(x_i, x_{-i}) + \frac{1}{2} u_i(x_i^*, x_{-i}^*) \\
        & = \frac{1}{2}u_i(x_i^*, x_{-i}) + \frac{1}{2}u_i(x_i, x_{-i}^*) = \int_S u_i(g(s)) d \pi_i(s)
    \end{align*}

    And so $u_i(x_i, x_{-i}^*) + u_i(x_i^*, x_{-i}) - u_i(x_i^*, x_{-i}^*) = u_i(x_i, x_{-i})$.
\end{proof}

\emph{Necessity of mutual independence of $X_i$ and $X_{-i}$.}

\begin{proof}
Suppose that $u_i(x) = u_{i,i}(x_i) + u_{i,-i}(x_{-i})$ for any $x \in X$. Then $u_{i,i}$ and $u_{i,-i}$ are measurable, since each is a section of $u_i$ plus a constant. So for any act $f \in F$,
\begin{align*}
    \int_S u_i(f(s))d\pi_i(s) = \int_{X_i}u_{i,i}(x_i)dP^{f,i}_i(x_i) + \int_{X_{-i}}u_{i,-i}(x_{-i})dP^{f,i}_{-i}(x_{-i}),
\end{align*}
which depends on $f$ only through $P^{f,i}_i$ and $P^{f,i}_{-i}$. So $f \sim_i g$ for any acts $f, g \in F$ for which $P^{f,i}_i = P^{g,i}_i$ and $P^{f,i}_{-i} = P^{g,i}_{-i}$.
\end{proof}

\subsection{Proof of Theorem \ref{theorem2}}

\emph{Sufficiency of restricted sovereignty.}

We begin by showing that the social belief $\pi_0$ is an affine combination of individual beliefs $\{\pi_i\}_{i \in I}$.

\emph{Aggregation of beliefs.}

\begin{proof}

Denote $\bm{\pi} = (\pi_i)_{i = 1}^n$ and $\hat{\bm{\pi}} = (\pi_i)_{i = 0}^n$. Let $\mathbf{Z}$ and $\hat{\mathbf{Z}}$ be the ranges of the vector measures $\bm{\pi}$ and $\hat{\bm{\pi}}$, respectively. Note that for any $\hat{\mathbf{z}} \in \hat{\mathbf{Z}}$, $\hat{\bm{\pi}}(S) - \hbz \in \hbbz$, because for an event $E$ such that $\hat{\bm{\pi}}(E) = \hbz$, $\hat{\bm{\pi}}(E^C) = \hat{\bm{\pi}}(S) -\hbz$. Since $\pi_0, \pi_1, \dots, \pi_n$ are countably additive and nonatomic, by the Lyapunov theorem, both $\mathbf{Z}$ and $\hbbz$ are convex.\footnote{The Lyapunov theorem states that the range of a nonatomic finite-dimensional vector measure is closed and convex. See \cite{liapounoff1940fonctions}, or \cite{ross2005elementary} for a simpler proof.}

\begin{claim} \label{claim.sub1}
    If $(z_0, \frac{1}{2}\bm\pi(S)) \in \hbbz$, then $z_0 = \frac{1}{2}$.
\end{claim}

\begin{proof}
    Suppose otherwise, that there exists an event $E \in \Sigma$ such that $\pi_i(E) = \frac{1}{2}$ for all $i \in I$, and $\pi_0(E) \neq \frac{1}{2}$. Since $\pi_i(E^C) = \frac{1}{2}$ for all $i \in I$ and $\pi_0(E^C) = 1 - \pi_0(E)$, we can assume without loss of generality, by replacing $E$ with $E^C$ if necessary, that $\pi_0(E) > \frac{1}{2}$.

    By minimal self-interest, there exists an $i$ for whom there exist $x_i, y_i \in X_i$ and $x_{-i} \in X_{-i}$ such that $f \succ_i g$, where $f(s) = (x_i, x_{-i})$ and $g(s) = (y_i, x_{-i})$ for all $s \in S$. Note that constant acts are lotteries, and so by restricted sovereignty, $f \succ_0 g$. And so $u_0(x_i, x_{-i}) > u_0(y_i, x_{-i})$.

    Consider acts $h, l \in F$, where $h^{-1}(x_i, x_{-i}) = E$, $h^{-1}(y_i, x_{-i}) = E^C$, $l^{-1}(x_i, x_{-i}) = E^C$, and $l^{-1}(y_i, x_{-i}) = E$. Since $\pi_0(E) > \frac{1}{2}$ and $u_0(x_i, x_{-i}) > u_0(y_i, x_{-i})$, $h \succ_0 l$. On the other hand, since $\pi_i(E) = \frac{1}{2}$, $h \sim_i l$.

    Since $h$ and $l$ are lotteries and $P^h_{-i} = P^l_{-i}$, this contradicts restricted sovereignty.

\end{proof}

The conclusion that social beliefs are an affine combination of individual beliefs follows from Proposition 3 in \cite{mongin1995consistent}. Reproduced below is a shorter proof by \cite{gilboa2004utilitarian}.

\begin{claim} \label{claim.sub2}
    For every $\mathbf{z} \in \mathbf{Z}$, there exists a unique $z_0 = z_0(\mathbf{z})$ such that $(z_0, \mathbf{z}) \in \hat{\mathbf{Z}}$.
\end{claim}

\begin{proof}
For existence, let $E$ be an event for which $\bm{\pi}(E) = \mathbf{z}$. Then $(\pi_0(E), \mathbf{z}) \in \hbbz$.

For uniqueness, suppose that $\hbz = (z_0, \mathbf{z})$ and $\hbz' = (z_0', \mathbf{z})$ are in $\hbbz$, where $z_0 < z_0'$. Note that $\hat{\bm{\pi}}(S) - \hat{\mathbf{z}}' \in \hat{\mathbf{Z}}$. By the convexity of $\hbbz$, $\frac{1}{2}\hbz + \frac{1}{2}(\hat{\bm{\pi}}(S) - \hbz') = (\frac{1}{2}z_0 + \frac{1}{2}(1-z_0'), \frac{1}{2}\bm{\pi}(S)) \in \hbbz$. Since $z_0 < z_0'$, $\frac{1}{2}z_0 + \frac{1}{2}(1-z_0') < \frac{1}{2}$, which contradicts Claim \ref{claim.sub1}.
\end{proof}

\begin{claim} \label{claim.sub3}
    For every $\mathbf{z}, \mathbf{z}' \in \mathbf{Z}$ and every $\beta \in [0,1]$, $z_0(\beta \mathbf{z} + (1-\beta)\mathbf{z}') = \beta z_0(\mathbf{z}) + (1-\beta)z_0(\mathbf{z}')$.
\end{claim}

\begin{proof}
    Let $\hbz = (z_0(\mathbf{z}), \mathbf{z})$ and $\hbz' = (z_0(\mathbf{z}'), \mathbf{z}')$. By the convexity of $\hbbz$, $\beta \hbz + (1-\beta) \hbz' \in \hbbz$. The first element is $\beta z_0(\mathbf{z}) + (1-\beta)z_0(\mathbf{z}')$. The rest of the elements are $\beta \mathbf{z} + (1-\beta) \mathbf{z}'$. The result follows from Claim \ref{claim.sub2}.
\end{proof}

    By Claim \ref{claim.sub3}, $z_0$ is affine on the convex set $\mathbf{Z}$, i.e. it preserves convex combinations. Note that $\mathbf{0} \in \mathbf{Z}$ and $z_0(\mathbf{0}) = 0$, since $\hat{\bm{\pi}}(\emptyset) = \mathbf{0}$. An affine function on a convex set extends to an affine function on the affine hull of the set, which here is the linear span of $\mathbf{Z}$, since $\mathbf{0} \in \mathbf{Z}$. Since $z_0(\mathbf{0}) = 0$, this extension is linear, and it extends further to a linear function on $\mathbb{R}^n$. So there exists $(\theta_i)_{i \in I} \in \mathbb{R}^n$ such that $z_0(\mathbf{z}) = \sum_{i\in I} \theta_i z_i$ for every $\mathbf{z} \in \mathbf{Z}$. Then, for any event $E$, $\pi_0(E) = \sum_{i \in I} \theta_i \pi_i(E)$. Substituting $E = S$ gives $\sum_{i \in I} \theta_i = 1$. The weights $(\theta_i)_{i \in I}$ are unique if and only if $\mathbf{Z}$ spans $\mathbb{R}^n$, i.e. if and only if $\pi_1, \dots, \pi_n$ are linearly independent.
\end{proof}

\emph{Aggregation of utility.}

\begin{proof}

Since $\sum_{i \in I}\theta_i = 1$, $\pi_0(E) = \sum_{i \in I}\theta_i\pi_i(E) = \pi_j(E)$ for any $E \in \Lambda$ and $j \in I$. So for any lottery $f \in \hat{F}$, society also believes that $f$ yields $P^f$, i.e. $\pi_0(\{s \in S: f(s) \in Y\}) = P^f(Y)$ for any measurable $Y \subseteq X$, and so
\begin{align*}
    \int_S u_0(f(s))d\pi_0(s) = \int_X u_0(x)dP^f(x).
\end{align*}

\cite{dubins1961cut} provide the following result for countably additive and nonatomic probability measures $(\pi_i)_{i\in I}$:
\begin{lemma} \label{dubspan}
    Let $p_1, \dots, p_m$ be non-negative numbers that sum to one. Then there exists a partition $E_1, \dots, E_m$ of $S$ such that, for all $1 \leq j \leq m$, and for all $i \in I$, $\pi_i(E_j) = p_j$.
\end{lemma}

Fix $i \in I$, and pick any $x^* = (x_i^*, x_{-i}^*) \in X$. We show that for any $x_i \in X_i$ and $x_{-i} \in X_{-i}$,
    \begin{align*}
        u_0(x_i, x_{-i}) = u_0(x_i, x_{-i}^*) + u_0(x_i^*, x_{-i}) - u_0(x_i^*, x_{-i}^*),
    \end{align*}
so that $u_0$ is the sum of a function of $x_i$ and a function of $x_{-i}$. By Lemma \ref{dubspan}, there exist lotteries $f,g \in \hat{F}$ such that $P^{f} = \left(\frac{1}{2}(x_i, x_{-i}), \frac{1}{2}(x_i^*, x_{-i}^*)\right)$ and $P^{g} = \left(\frac{1}{2}(x_i^*, x_{-i}), \frac{1}{2}(x_i, x_{-i}^*)\right)$. By mutual independence, $f \sim_i g$. By restricted sovereignty, $f \sim_0 g$. Since $\succsim_0$ is SEU-rational, and society also believes that $f$ and $g$ yield $P^f$ and $P^g$,
    \begin{align*}
        & \int_S u_0(f(s)) d \pi_0(s) = \frac{1}{2}u_0(x_i, x_{-i}) + \frac{1}{2} u_0(x_i^*, x_{-i}^*) \\
        & = \frac{1}{2}u_0(x_i^*, x_{-i}) + \frac{1}{2}u_0(x_i, x_{-i}^*) = \int_S u_0(g(s)) d \pi_0(s),
    \end{align*}
which gives the display above.

Since the fixed $i$ was arbitrary, for each $i \in I$, $u_0$ is the sum of a function of $x_i$ and a function of $x_{-i}$. Then, as in the proof of Theorem \ref{theorem1}, there exist functions $\{u_{0,i}\}_{i \in I}$ such that $u_0(x) = \sum_{i \in I} u_{0,i}(x_i)$.

For each $i$, let $u_{i,i}$ be a subutility index in the utility index $u_i = u_{i,i} + u_{i,-i}$. The fact that $u_{0,i}$ is a positive affine transformation of $u_{i,i}$ is proven similarly to Theorem \ref{theorem1}: for any $P_i, Q_i \in \Delta(X_i)$ and $\bar{x}_{-i} \in X_{-i}$, Lemma \ref{dubspan} gives lotteries $f, g \in \hat{F}$ for which $P^f$ and $P^g$ are the products of $P_i$ and $Q_i$, respectively, with the degenerate lottery on $\bar{x}_{-i}$, so that $P^f_{-i} = P^g_{-i}$.
\end{proof}

\emph{Necessity of restricted sovereignty.}

\begin{proof}
Suppose that $u_0 = \sum_{j \in I}w_ju_{j,j}$, where each $w_j > 0$, and that $\pi_0 = \sum_{j \in I}\theta_j\pi_j$, where $\sum_{j \in I}\theta_j = 1$. As in the proof of sufficiency, $\int_S u_0(f(s))d\pi_0(s) = u_0(P^f)$ for any lottery $f \in \hat{F}$, and similarly $\int_S u_i(f(s))d\pi_i(s) = u_i(P^f)$ for each $i \in I$. Consider any $i \in I$ and lotteries $f, g \in \hat{F}$ for which $P^f_{-i} = P^g_{-i}$. Since $P^f_j = P^g_j$ for all $j \neq i$,
\begin{align*}
    f \succsim_0 g & \iff \sum_{j \in I}w_ju_{j,j}(P^f_j) \geq \sum_{j \in I}w_ju_{j,j}(P^g_j) \iff u_{i,i}(P^f_i) \geq u_{i,i}(P^g_i) \\
    & \iff u_{i,i}(P^f_i) + u_{i,-i}(P^f_{-i}) \geq u_{i,i}(P^g_i) + u_{i,-i}(P^g_{-i}) \iff f \succsim_i g,
\end{align*}
where the second and third equivalences hold since $P^f_{-i} = P^g_{-i}$ and $w_i > 0$.
\end{proof}

\end{document}